\documentclass[12pt]{article}

\usepackage[margin=1in]{geometry}

\usepackage{enumitem} 

\usepackage{centernot}

\usepackage{times}

\usepackage{setspace} 
\usepackage{caption}
\renewcommand{\arraystretch}{0.6} 

\usepackage{comment}

\usepackage{amsmath}
\usepackage{amssymb}
\usepackage{amsthm} 
\usepackage{mathtools} 
\usepackage{bbm} 
\newcommand{\ind}{\perp\!\!\!\!\perp} 
\usepackage{centernot} 

\usepackage{dcolumn}
\newcolumntype{d}[1]{@{}D{.}{.}{#1}@{}} 

\usepackage[utf8]{inputenc} 
\usepackage[english]{babel}

\usepackage{graphicx} 

\usepackage{float}

\usepackage{array} 
\usepackage{arydshln}

\usepackage{multirow}

\newtheorem{assumption}{Assumption}
\newtheorem{proposition}{Proposition}

\newtheorem{lemma}{Lemma}

\usepackage{appendix}

\numberwithin{equation}{section}

\newcommand\blfootnote[1]{
  \begingroup
  \renewcommand\thefootnote{}\footnote{#1}
  \addtocounter{footnote}{-1}
  \endgroup
}

\usepackage[hang]{footmisc}
\usepackage{hyperref}
\hypersetup{colorlinks=true, allcolors=black}

\usepackage[authoryear,round,sort,comma,longnamesfirst]{natbib} 
\usepackage{myapalike}
\begin{document}

\begin{titlepage}
\begin{center}
\linespread{1.2}
\Large{\textbf{Identification by non-Gaussianity in structural smooth transition vector autoregressive models}}\\

\vspace{0.5cm} 
\Large{Savi Virolainen}\\
\large{University of Helsinki}\\
\vspace{1.0cm}


\begin{abstract}
\noindent We show that structural smooth transition vector autoregressive models are statistically identified if the shocks are mutually independent and at most one of them is Gaussian. This extends a known identification result for linear structural vector autoregressions to a time-varying impact matrix. 
We also propose an estimation method, show how a blended identification strategy can be adopted to address weak identification, and establish a sufficient condition for ergodic stationarity. The introduced methods are implemented in the accompanying R package sstvars. Our empirical application finds that a positive climate policy uncertainty shock reduces production and raises inflation under both low and high economic policy uncertainty, but its effects, particularly on inflation, are stronger during the latter. 

\end{abstract}


\vfill

\blfootnote{The author thanks Markku Lanne for the helpful discussions and comments as well as the Research Council of Finland for the financial support (Grant 347986).}
\blfootnote{Contact address: Savi Virolainen, Faculty of Social Sciences, University of Helsinki, P. O. Box 17, FI–00014 University of Helsinki, Finland; e-mail: savi.virolainen@helsinki.fi. ORCiD ID: 0000-0002-5075-6821.}
\blfootnote{The author does not have conflicts of interest to declare.}

\end{center}
\end{titlepage}

\section{Introduction}
Linear structural vector autoregressive (SVAR) models have become a standard tool in empirical macroeconomics, as they can effectively capture constant dynamic relationships among the variables and facilitate tracing out the causal effects of economic shocks. Macroeconomic systems may, however, exhibit variation in their dynamics, induced by crises, policy shifts, or business cycle fluctuations, for example. Hence, also the effects of the shocks may depend on the state of the economy. Such features cannot be accommodated by linear SVAR models, and therefore nonlinear SVAR models that are able to capture such features are often employed \citep[see, e.g.,][Chapter~18]{Kilian+Lutkepohl:2017}. In nonlinear SVAR models, the structural shocks are typically identified with conventional methods that impose economically interpretable but restrictive assumptions, such as zero contemporaneous interactions among some of the variables, which are not always plausible in practice. To address this issue, statistical identification methods relying on the statistical properties of the data can be used \citep[see][Chapter~14]{Kilian+Lutkepohl:2017}. 



There are two main branches in the statistical identification literature: identification by heteroskedasticity
(\citealp{Rigobon:2003}, \citealp{Lanne+Lutkepohl:2010}, \citealp{Bacchiocchi+Fanelli:2015}, \citealp{Lutkepohl+Netsunajev:2017}, \citealp{Lewis:2021}, \citealp{Virolainen:2025}, and others) and identification by non-Gaussianity (\citealp{Lanne+Meitz+Saikkonen:2017},  \citealp{Lanne+Luoto:2021}, \citealp{Lanne+Liu+Luoto:2023}, and others). To the best of our knowledge, this paper is the first to examine identification by non-Gaussianity in nonlinear SVAR models. Under certain statistical conditions, both types of statistical information typically suffice to identify the shocks in a linear SVAR model. As pointed out by \cite{Lutkepohl+Netsunajev:2017}, among others, identification by heteroskedasticity without additional restrictions, however, has the major drawback in nonlinear SVAR models that it imposes time-invariant (relative) impact effects of the shocks. We show that identification by non-Gaussianity, in turn, facilitates recovering the shocks without such undesirable restrictions. 


This paper contributes to the literature on identification by non-Gaussianity by extending the framework of \cite{Lanne+Meitz+Saikkonen:2017} to structural smooth transition vector autoregressive (STVAR) models, which is a major class of nonlinear SVAR models \citep[see, e.g.,][]{Hubrich+Terasvirta:2013}. The STVAR model can flexibly capture nonlinear data generating dynamics by accommodating multiple regimes and gradual as well as abrupt shifts between them, governed by the transition weights. In contrast to its linear counterpart, the impact matrix of the structural STVAR model should generally allow for time-variation in the impact responses of the variables to the shocks, which complicates identification. Nevertheless, similarly to \cite{Lanne+Meitz+Saikkonen:2017}, it turns out that identification is achieved when the shocks are mutually independent and at most one of them is Gaussian. We show that under this condition, the shocks are readily identified up to ordering and signs when the impact matrix of the STVAR model is defined as a weighted sum of the impact matrices of the regimes. While we focus on exogenous, logistic, and threshold transition weights, our results can be extended to other suitable weight functions as well. 


In line with the statistical identification literature, external information is required to label the identified structural shocks as economic shocks. Our nonlinear setup has the additional complication that the same shock should be assigned to the same column of the impact matrix across all regimes. Moreover, our experience shows that identification is often weak in the sense that there are often multiple local solutions with a fit close to the global solution. We find such local solutions to be often largely quite similar to each other, with the differences frequently but not exclusively related to different ordering and signs of the columns of the regime-specific impact matrices. Nonetheless, since different local solutions may produce different results in structural analysis, we recommend addressing weak identification by adopting a blended identification strategy that combines identification by non-Gaussianity with supplementary identifying information \citep[cf.][]{Carriero+Marcellino+Tornese:2024}.

Our methods are demonstrated in an empirical application to the macroeconomic effects of climate policy uncertainty (shocks) that considers monthly U.S. data from 1987:4 to 2024:12. 
Following \cite{Khalil+Strobel:2023} and \cite{Huang+Punzi:2024}, we measure climate policy uncertainty (CPU) with the CPU index of \cite{Gavriilidis:2021}, which is constructed based on the amount of newspaper coverage on topics related to CPU. We are interested in studying how the effects of the CPU shock vary depending on the level of economic policy uncertainty (EPU). Therefore, we fit a two-regime structural logistic STVAR model using the first lag of the EPU index \citep{Baker+Bloom+Davis:2016} as the switching variable. 
We find that a positive CPU shock reduces production and raises inflation in times of both low and high EPU, but its effects, particularly on inflation, are stronger in the periods of high EPU. Our results are, hence, in line with the previous literature suggesting that a positive CPU shock reduces production and raises inflation (\citealp{Khalil+Strobel:2023} and \citealp{Huang+Punzi:2024}), while \cite{Fried+Novan+Peterman:2022} found that it decreases production.

The rest of this paper is organized as follows. Section~\ref{sec:stvar} presents the framework of reduced form STVAR models. 
Section~\ref{sec:structstvar} discusses identification of the shocks in structural STVAR models, presents our identification results, and discusses the problem of labeling the shocks as well as how to address weak identification by adopting a blended identification strategy. 
Section~\ref{sec:estimation} discusses stationarity of the model and proposes estimating its parameters with a penalized likelihood-based estimator using a three-step procedure. Section~\ref{sec:empapp} presents the empirical application and Section~\ref{sec:conclusions} concludes. Further details can be found in the appendices, and the introduced methods have been implemented to the accompanying R package sstvars \citep{sstvars}, which is available via the CRAN repository.


\section{Reduced form STVAR models}\label{sec:stvar}

Let $y_t$, $t=1,2,...$, be the $d$-dimensional time series of interest and $\mathcal{F}_{t-1}$ denote the $\sigma$-algebra generated by the random vectors $\lbrace y_{t-j}, j>0 \rbrace$ (i.e., $\mathcal{F}_{t-1}$ contains the information about the history of the process). We consider STVAR models with $M$ regimes and autoregressive order $p$ assumed to satisfy
\begin{align}
y_t &=\sum_{m=1}^M \alpha_{m,t}\mu_{m,t} + u_t, \quad u_{t} \sim MD(0, \Omega_{y,t},\nu),\label{eq:stvar1}\\
\mu_{m,t} &= \phi_{m} + \sum_{i=1}^{p}A_{m,i}y_{t-i}, \quad m=1,...,M,\label{eq:stvar2}
\end{align}
where $\phi_{1},...,\phi_{M}\in\mathbb{R}^{d}$ are the intercept parameters; $A_{1,i},...,A_{M,i}\in\mathbb{R}^{d\times d}$, $i=1,...,p$, are the autoregression matrices; $u_t$ is a martingale difference sequence of reduced form innovations; $\Omega_{y,t}$ is the positive definite conditional covariance matrix of $u_t$ (conditional on $\mathcal{F}_{t-1}$), which may depend on $\alpha_{1,t},...,\alpha_{M,t}$ and $y_{t-1},...,y_{t-p}$; and $\nu$ collects any other parameters that the distribution of $u_t$ depends on. 

The transition weights $\alpha_{m,t}$ are assumed to be either exogenous (nonrandom) or $\mathcal{F}_{t-1}$-measurable functions of $\lbrace y_{t-j}, j=1,...,p \rbrace$, and to satisfy $\sum_{m=1}^{M}\alpha_{m,t}=1$ at all $t$. Thus, we accommodate the types of transition weights typically used in macroeconomic applications, including exogenous, logistic, and threshold weights discussed below. The transition weights express the proportions of the regimes the process is in at each point of time and determine how the process shifts between them. 

It is easy to see that, conditional on $\mathcal{F}_{t-1}$, the conditional mean of the above-described process is $\mu_{y,t} \equiv E[y_t|\mathcal{F}_{t-1}] = \sum_{m=1}^M \alpha_{m,t}\mu_{m,t}$, a weighted sum the regime-specific means $\mu_{m,t}$ with the weights given by the transition weights $\alpha_{m,t}$. The discussion on the specification of the conditional covariance matrix $\Omega_{y,t}$ is postponed to Section~\ref{sec:structstvar}. The linear vector autoregressive (VAR) model is obtained as a special case by assuming constant autoregressive dynamics across the regimes, i.e., $A_{m,i}=A_i$ and $\phi_{m}=\phi$ for all $i=1,...,p$ and $m=1,...,M$.\footnote{See \cite{Hubrich+Terasvirta:2013} for a survey on STVAR literature, including also specifications more general than ours.}

A popular specification is obtained by assuming two regimes ($M=2$) and logistic transition weights. In the logistic STVAR model \citep{Anderson+Vahid:1998}, the transition weights vary according to a logistic function as
\begin{equation}\label{eq:alpha_mt_logistic}
\alpha_{1,t} = 1 - \alpha_{2,t} \ \ \text{and} \ \ \alpha_{2,t} = [1 + \exp\lbrace -\gamma(z_t - c)\rbrace ]^{-1},
\end{equation}
where $c\in\mathbb{R}$ is the location parameter, $\gamma > 0$ is the scale parameter, and $z_t$ is the switching variable, which we assume to be a lagged endogenous variable up to the lag $p$ (that is, $z_t\in\{y_{it-j}, i=1,...,d, j=1,...,p \}$). 
 The location parameter $c$ determines the midpoint of the transition function, i.e., the value of the switching variable when the weights are equal. The scale parameter $\gamma$, in turn, determines the smoothness of the transitions (the smaller $\gamma$ is, the smoother the transition is), and it is assumed strictly positive so that $\alpha_{2,t}$ is increasing in $z_t$. In the special case when $\gamma\rightarrow \infty$, regime-switches are discrete and the logistic weights reduce to the threshold weights of \cite{Tsay:1998} (see Appendix~\ref{sec:identtvar} for details).

\section{Structural STVAR models identified by non-Gaussianity}\label{sec:structstvar}
\subsection{Structural STVAR models}\label{sec:structstvarsub}

A structural STVAR model is obtained from the reduced form model defined in Section~\ref{sec:stvar} by identifying the serially and mutually uncorrelated structural shocks $e_{t}=(e_{1t},...,e_{dt})$ $(d\times 1)$ from the reduced form innovations $u_t$. Specifically, the structural shocks are recovered from the reduced form innovations with the transformation 
\begin{equation}
e_t=B_{y,t}^{-1}u_t,
\end{equation}
where $B_{y,t}$ is an invertible ($d\times d$) impact matrix that governs the contemporaneous relationships of the shocks and may depend on $\alpha_{1,t},...,\alpha_{M,t}$ and $y_{t-1},...,y_{t-p}$. In other words, assuming a unit variance normalization for the structural shocks, the identification problem amounts to finding an impact matrix $B_{y,t}$ such that the conditional covariance matrix $\Omega_{y,t}=B_{y,t}B_{y,t}'$. However, $B_{y,t}$ is not generally uniquely identified without further restrictions, as many such decompositions exist. Various identification methods have been proposed to resolve this issue \citep[see, e.g.,][]{Kilian+Lutkepohl:2017}.

\cite{Lanne+Meitz+Saikkonen:2017}, among others, have shown that if the shocks are mutually independent and at most one of them is Gaussian, a static impact matrix is readily identified (up to ordering and signs of its columns) without additional restrictions. We extend this result to a time-varying impact matrix, and therefore, it is useful to parametrize the structural model directly with $B_{y,t}$. 
%
Hence, we make the identity $u_t=B_{y,t}e_t$ explicit in Equation~(\ref{eq:stvar1}) as
\begin{equation}\label{eq:stvarstruct}
y_t =\sum_{m=1}^M \alpha_{m,t}\left(\phi_{m} + \sum_{i=1}^{p}A_{m,i}y_{t-i} \right) + B_{y,t}e_t, \quad e_t\sim IID(0,I_d,\nu),
\end{equation}
where $e_t$ are independent and identically distributed structural errors with identity covariance matrix and a distribution that may depend on the parameter $\nu$. 

To incorporate time-variation in the impact matrix, we specify its functional form. A natural specification for structural STVAR models is to assume a constant impact matrix for each of the regimes and define the impact matrix of the process to be their weighted sum as
\begin{equation}\label{eq:Bt}
B_{y,t} = \sum_{m=1}^M \alpha_{m,t}B_m,
\end{equation}
where $B_1,...,B_M$ are invertible $(d\times d)$ impact matrices of the regimes. Also the impact matrix $B_{y,t}$ needs to be invertible to ensure positive definiteness of the conditional covariance matrix, which does not automatically follow from the invertibility of $B_1,...,B_M$. Nevertheless, it turns out that $B_{y,t}$ defined in~(\ref{eq:Bt}) is invertible for all $t$ almost everywhere in $[B_1:...:B_M] \in\mathbb{R}^{d\times dM}$, as is stated in the following lemma (which is proven in Appendix~\ref{sec:proofinvertibility}).

\begin{lemma}\label{lemma:invertibility}
Suppose $B_1,...,B_M$ are $(d\times d)$ matrices and $\alpha_{1,t},...,\alpha_{M,t}$ scalars such that $\sum_{m=1}^M\alpha_{m,t}=1$ and $\alpha_{m,t}\geq 0$ for all $t$ and $m=1,...,M$. Then, the matrix $B_{y,t}=\sum_{m=1}^M \alpha_{m,t}B_m$ is invertible for all $t$ almost everywhere in $[B_1:...:B_M] \in\mathbb{R}^{d\times dM}$.
\end{lemma} 
The result of Lemma~\ref{lemma:invertibility} holds almost everywhere in $[B_1:...:B_M] \in\mathbb{R}^{d\times dM}$, which means that the set where the matrices $B_1,...,B_M$ are such that $B_{y,t}$ is singular for some $t$ has Lebesgue measure zero. In other words, invertible matrices $B_1,...,B_M$ generally lead to an invertible impact matrix $B_{y,t}$, excluding some special cases of $B_1,...,B_M$ for which this result does not hold (e.g., if $B_2=-B_1$, then $B_{y,t}=0$ when $\alpha_{1,t}=\alpha_{2,t}=0.5$).
To support intuition, note that the result of Lemma~\ref{lemma:invertibility} implies that if the matrices $B_1,...,B_M$ are drawn by random from a continuous distribution, the probability of obtaining matrices such that $B_{y,t}$ is singular for some $t$ is zero. 

Assuming an impact matrix of the form~(\ref{eq:Bt}), the conditional covariance matrix of $y_t$, conditional on $\mathcal{F}_{t-1}$, is obtained as \begin{equation}\label{eq:condcovmat}
\Omega_{y,t}=\left(\sum_{m=1}^M \alpha_{m,t}B_m \right) \left(\sum_{m=1}^M \alpha_{m,t}B_m \right)' = \sum_{m=1}^M \alpha_{m,t}^2\Omega_m +  \sum_{m=1}^M\sum_{n=1,n\neq m}^M\alpha_{m,t}\alpha_{n,t}\Omega_{m,n},
\end{equation}
where $\Omega_m\equiv B_mB_m'$ and $\Omega_{m,n}\equiv B_mB_n'$. The conditional covariance matrix of $y_t$ can thus be described as a weighted sum of $M^2$ matrices with the weights varying in time according to the transition weights $\alpha_{m,t}$, $m=1,...,M$. If the process is completely in one of the regimes, i.e., $\alpha_{m,t}=1$ for some $m\in\lbrace 1,...,M\rbrace$, then $\Omega_{y,t}$ reduces to $\Omega_{m}$, implying that this constitutes the conditional covariance matrix of Regime~$m$. When the process is not completely in any of the regimes, $\Omega_{y,t}$ depends on both the regime-specific covariance matrices and the cross terms. Hence, this specification of the structural model is different to the conventional reduced form STVAR specification in which the conditional covariance matrix of $y_t$ is a weighted sum of the covariance matrices of the regimes. Nevertheless, our model specification facilitates exploiting non-Gaussianity of the shocks in their identification as is shown in the next section. 

\subsection{Identification of the shocks by non-Gaussianity}\label{sec:ident_nongaus}

The key identification assumption is that the structural shocks $e_t=(e_{1t},...,e_{dt})$ are mutually independent and at most one of them is Gaussian. We also assume that $e_t$ is an IID sequence with zero mean and identity covariance matrix as is formally stated in the following assumption. 
\begin{assumption}\label{as:shocks}
\begin{enumerate}[label=(\roman*)]
\item The structural error process $e_t=(e_{1t},...,e_{dt})$ is a sequence of independent and identically distributed random vectors such that $e_t$ is independent of $\mathcal{F}_{t-1}$ and each component $e_{it}$, $i=1,...,d$, follows a continuous distribution with zero mean, unit variance, and a density that is strictly positive almost everywhere in $\mathbb{R}$.\label{cond:shockiid}
\item The components of $e_t=(e_{1t},...,e_{dt})$ are mutually independent and at most one of them has a Gaussian marginal distribution.
\end{enumerate}
\end{assumption}
The assumption of a unit variance is merely a normalization since the variance of the shocks is captured by the impact matrix $B_{y,t}$. Also, due to the time-variation of the impact matrix, Assumption~\ref{as:shocks} does not imply that the reduced form innovations $u_t$ are independent nor that they are identically distributed. The assumption of strictly positive density on sets of positive (Lebesgue) measure 
mainly rules out bounded shock distributions. 


It is helpful to distinguish two layers of identification. First, for any fixed $t$ and conditional on $\mathcal{F}_{t-1}$, the impact matrix $B_{y,t}$ in~\eqref{eq:stvarstruct} is identified up to column permutations and sign changes by independent component analysis (ICA) under Assumption~\ref{as:shocks} \citep[cf.][]{Lanne+Meitz+Saikkonen:2017}. Second, under the structural parametrization $B_{y,t}=\sum_{m=1}^M \alpha_{m,t} B_m$ in~\eqref{eq:Bt}, identification of the regime-specific matrices $B_1,...,B_M$ (up to column ordering and signs) additionally requires identification of the transition weight parameters, or in the case of exogenous weights, sufficient variation in the known weights. The next lemma, partly similar to Proposition~1 in \cite{Lanne+Meitz+Saikkonen:2017}, therefore formalizes the per-$t$ ICA result, and, in the case of logistic weights~(\ref{eq:alpha_mt_logistic}) establishes identification of the weight function parameters (and, for completeness, identification of the AR parameters for logistic and exogenous weights). Specific weight functions are assumed here for simplicity, but the result can be extended to other suitable weight functions as well. Identification of $B_1,...,B_M$ then follows from these ingredients.
%
%
\begin{lemma}\label{lemma:Bt_at_each_t}
Consider the STVAR model defined in Equations~(\ref{eq:stvarstruct}) and (\ref{eq:Bt}) with Assumption~\ref{as:shocks} satisfied. Then, conditionally on $\mathcal{F}_{t-1}$, at each $t$, $B_{y,t}$ is uniquely identified up to ordering and signs of its columns almost everywhere in $[B_1:...:B_M] \in \mathbb{R}^{d\times dM}$.

Moreover, suppose the transition weights are either exogenous (nonrandom) or follow the logistic process defined in~(\ref{eq:alpha_mt_logistic}) (with $M=2$ assumed), 
and that the following conditions hold:
\begin{enumerate}[label=(\arabic*)]
\item $(\phi_{m},\text{vec}(A_{m,1}),...,\text{vec}(A_{m,p})) \neq (\phi_{n},\text{vec}(A_{n,1}),...,\text{vec}(A_{n,p}))$ for all $m\neq n \in \lbrace 1,...,M\rbrace$, and\label{cond:different_arpars}
\item (for exogenous weights) there exists $M$ indices $t_1,...,t_M\in\lbrace 1,....,T\rbrace$ such that the corresponding coefficient vectors $(\alpha_{1,t_i},...,\alpha_{M,t_i})$, $i=1,...,M$, are linearly independent.\label{cond:linind}
\end{enumerate}
Then, the parameters $\phi_{m},A_{m,1},...A_{m,p}$, $m=1,...,M$, and for logistic models $c$ and $\gamma$, are uniquely identified. 
\end{lemma}

Lemma~\ref{lemma:Bt_at_each_t} is proven in Appendix~\ref{sec:prooflemmaBteach}. Condition~\ref{cond:different_arpars} implies distinguishable regimes, whereas Condition~\ref{cond:linind} guarantees a certain small amount of variation in exogenous weights. Lemma~\ref{lemma:Bt_at_each_t} concludes that at each $t$, conditionally on $\mathcal{F}_{t-1}$, the impact matrix $B_{y,t}$ is unique up to ordering and signs of its columns under Assumption~\ref{as:shocks}. That is, at each $t$‚ changing the ordering or signs of the columns of $B_{y,t}$ would lead to an observationally equivalent model, but changing $B_{y,t}$ in any other way would lead to an observationally distinct model.
It follows that if the impact matrix is time-invariant as in \cite{Lanne+Meitz+Saikkonen:2017}, i.e., $B_{y,t}=B$ for some constant matrix $B$, the structural shocks are identified up to ordering and signs. 
However, when the impact matrix varies over time, two complications arise. 

First, because $B_{y,t}$~(\ref{eq:Bt}) is identified only up to column ordering and signs at each $t$, its unique identification requires constraints on the regime-specific impact matrices $B_1,...,B_M$ such that any reordering or sign changes in their columns lead to observationally distinct $B_{y,t}$ at some $t$. Second, since $B_{y,t}$ is not a matrix of constant parameters but a function of parameters, it needs to be shown that the parameters in its functional form are identified. Also, in addition to $B_1,...,B_M$, the impact matrix $B_{y,t}$ depends on the transition weights $\alpha_{m,t}$, $m=1,...,M$, which must therefore be identified as well. To that end, we assume either the logistic or exogenous transition weights considered in Lemma~\ref{lemma:Bt_at_each_t}.

The following proposition, proven in Appendix~\ref{sec:proofstvarident}, establishes unique identification of $B_1,...,B_M$.
\begin{proposition}\label{prop:stvar_ident}
Consider the STVAR model defined in Equations~(\ref{eq:stvarstruct}) and (\ref{eq:Bt}) with Assumption~\ref{as:shocks} and the conditions of Lemma~\ref{lemma:Bt_at_each_t} satisfied, and suppose the transition weights are either exogenous (nonrandom) or follow the logistic process defined in~(\ref{eq:alpha_mt_logistic}) with $\gamma\in(0,\infty)$ (and $M=2$ assumed). Moreover, suppose 
the following conditions hold:
\begin{enumerate}[label=(\alph*)]
\item the ordering and signs of the columns of $B_1$ are fixed, and\label{cond:stvar_B_1}
\item (for exogenous weights) the vector $(\alpha_{1,t},...,\alpha_{M,t})$ takes a value for some $t$ such that none of its entries is zero.\label{cond:stvar_alpha}
\end{enumerate}
Then, the matrices $B_1,...,B_M$ are uniquely identified almost everywhere in $[B_1:...:B_M] \in\mathbb{R}^{d\times dM}$.
\end{proposition}
Proposition~\ref{prop:stvar_ident} essentially states that if the structural shocks are independent and at most one of them is Gaussian (Assumption~\ref{as:shocks}), the impact matrices $B_1,...,B_M$, and hence, the structural shocks are identified. The result holds for almost every $[B_1:...:B_M] \in\mathbb{R}^{d\times dM}$ because the identification may fail in some special cases, but this set has Lebesgue measure zero.\footnote{Beyond possible singularity of the impact matrix $B_{y,t}$ (see Lemma~\ref{lemma:Bt_at_each_t}), the failure of the identification in a measure zero set only concerns the possibility that fixing the ordering and signs of the columns of $B_1$ does not fix the ordering and signs of the columns of $B_2,...,B_M$. 
} 

Condition~\ref{cond:stvar_B_1} of Proposition~\ref{prop:stvar_ident} states that the ordering and signs of the columns of $B_1$ should be fixed (e.g., by assuming that the first nonzero entry in each column is positive and the first nonzero entries are in a decreasing order, given that none of them are equal), which fixes the ordering and signs of the columns of $B_{y,t}$ for all $t$. If the shocks follow, for example, the skewed $t$-distributions defined in Equation~(\ref{eq:skew_t_dens}) (in Section~\ref{sec:estimation} discussing estimation), any fixed ordering and signs of the columns of $B_1$ can be assumed without loss of generality. This is because reordering the columns would just reorder the shocks and swapping a sign of a column corresponds to swapping the sign of the related shock and skewness parameter. 
Condition~\ref{cond:stvar_alpha} states that exogenous transition weights should, for some $t$, be strictly positive for all the regimes, which is used to establish that Condition~\ref{cond:stvar_B_1} fixes the ordering and signs of the columns of $B_{y,t}$ for all $t$. For logistic weights, this is achieved via the assumption $\gamma<\infty$, guaranteeing that the regime-switches are not discrete.\footnote{To accommodate discrete regime-switches, which are obtained as a special case of the logistic weights~(\ref{eq:alpha_mt_logistic}) with the smoothness parameter tending to infinity, we establish the identification of the shocks in the threshold VAR model of \cite{Tsay:1998} in Appendix~\ref{sec:identtvar}.}


Somewhat surprisingly, our experience shows that identification appears to be often weak in the sense that there are multiple local maxima of the (penalized) log-likelihood function (discussed in Section~\ref{sec:estimation}) with (penalized) log-likelihoods close to each other and to the global maximum. The parameter values related to each such local maximum seem to be often largely quite similar to each other, but have some differences, frequently but not exclusively related to different ordering and signs of the columns of $B_2,...,B_M$. Since such different local solutions may produce different results in structural analysis, this weak identification should be appropriately addressed, for instance, by combining Proposition~\ref{prop:stvar_ident} with supplementary identifying information as discussed in the next section.

\subsection{Labeling the shocks and blended identification}\label{sec:labellingshocks}
As in the linear SVAR model of \cite{Lanne+Meitz+Saikkonen:2017}, the statistically identified structural shocks do not necessarily have economic interpretations, and labeling them as economic shocks requires external information. However, labeling the shocks based on the estimates of the impact matrices $B_1,...,B_M$ might not always be straightforward, as the same shock must be associated with the same column of the impact matrix $B_m$ in all regimes. For example, if a positive supply shock should increase output and decrease prices on impact in all regimes, labeling the $i$th shock as the supply shock requires the $i$th column of all $B_1,...,B_M$ to satisfy such signs. Moreover, as discussed in Section~\ref{sec:ident_nongaus}, our experience shows that identification is often weak in the sense that there are multiple local solutions with fit close to the global optimum. While we find such local solutions to be often largely quite similar to each other, they may produce different impulse response functions even when the shock of interest is associated with the same column of all $B_1,...,B_M$. 

To formally address weak identification and the problem of labeling the shocks, we recommend employing a "blended identification" strategy that combines identification by non-Gaussianity with supplementary identifying information \citep[cf.][]{Carriero+Marcellino+Tornese:2024}. Specifically, we propose imposing overidentifying restrictions that are sufficient to yield a unique local solution (among the ones with fit close to the global solution) and facilitate labeling the shocks of interest. Since different local solutions may yield different results, the restrictions should be economically reasonable and serve to exclude less plausible alternative solutions. 

As a simple example, in a bivariate system of output and prices, it may be reasonable to assume that for each variable, one of the shocks has greater impact effect on that variable than the other shock, and that this effect has the same sign across regimes. If one of the shocks is to be interpreted as a demand shock, it might be useful to further assume that the shock with the largest effect on output also moves prices in the same direction in all regimes. Similarly, if the other shock is intended to represent a supply shock, it could be useful to assume that it moves output and prices in the opposite directions in all regimes.\footnote{
Also other forms of information can be incorporated into the blended identification strategy, such as narrative restrictions and zero impact effect restrictions (the latter of which are testable due to statistical identification), for instance.
} 
See our empirical application in Section~\ref{sec:empapp} for an example that involves more variables.

As a practical consideration, note that our estimation method, described in Section~\ref{sec:estimation} and implemented to the accompanying R package sstvars \cite{sstvars}, produces a set of alternative local solutions. 
By filtering these local solutions based on whether they satisfy the imposed overidentifying restrictions, it is straightforward to determine which restrictions are sufficient in the considered specification to exclude all but one of the local maximums (among the ones with a fit close to the presumed global optimum). Finally, to label the shocks of interest, the imposed restrictions can often be used to uniquely associate each of them to the same column of the impact matrix across all regimes. 

\section{Estimation}\label{sec:estimation}

\subsection{Maximum likelihood estimation and ergodic stationarity}\label{sec:ml_stat}
The parameters of the structural STVAR model discussed in Section~\ref{sec:structstvar} can be estimated by the method of maximum likelihood (ML). 
To obtain a well-behaving estimator with desirable asymptotic properties such as consistency, the parameter space is often restricted to the region where the model is ergodic stationary. Building on the results of \cite{Saikkonen:2008} \citep[see also][]{Kheifets+Saikkonen:2020}, it can be shown that when the transition weights $\alpha_{m,t}$ are logistic~\eqref{eq:alpha_mt_logistic} or of the threshold form, a sufficient condition for ergodic stationarity is that the joint spectral radius of the companion form AR matrices of the regimes is strictly less than one (see Theorem~\ref{thm:stat} in Appendix~\ref{sec:stat}).\footnote{\label{footnote:B1B2cond} For logistic models (with $M=2$ assumed), we additionally require that the matrix $B_1^{-1}B_2$ has no negative real eigenvalues to rule out singular convex combinations of $B_1$ and $B_2$ (see Appendices~\ref{sec:stat} and~\ref{sec:proofthmstat} for details). See Appendix~\ref{sec:identtvar} for the definition of the threshold weights.} Because this condition is computationally demanding to verify \citep[e.g.,][]{Chang+Blondel:2013}, it is not particularly useful for restricting the parameter space during estimation. Therefore, we instead impose the usual stability condition for each regime (Condition~\ref{cond:necessary} in Appendix~\ref{sec:stat}), as it is a necessary condition for the sufficient one. The sufficient condition can then be checked after estimation for the solutions of interest. 

Maximizing the log-likelihood function can be challenging in practice due to its high multimodality, induced by the nonlinear dynamics.\footnote{The log-likelihood function is presented in Section~\ref{sec:loglik} for shocks that follow skewed $t$-distributions.} Imposing the stability condition can make estimation particularly difficult when the data is persistent, as in such cases numerical optimization algorithms frequently gravitate to the boundary of the parameter space, where they perform poorly. 
To overcome this issue, we propose to allow for unstable estimates and to maximize the penalized log-likelihood function, obtained by adding a penalty term to the log-likelihood function that penalizes parameter values falling outside or close to the boundary of the stability region. In this way, the optimization algorithm can explore the parameter space also outside the stability region, facilitating improved performance, while sufficient penalization eventually steers the algorithm back to the stability region.\footnote{Penalized likelihood-based estimation has been previously applied to time series models by \cite{Nielsen+Rahbek:2024}. They impose nonnegativity constraints in the autoregressive conditional heteroskedasticity model, allowing the estimation algorithm to explore the boundary of the parameter space. \cite{Nielsen+Rahbek:2024} use penalization also to determine which parameters should be set to zero.} 

\subsection{The penalized log-likelihood function}\label{sec:loglik}

Before introducing the penalized log-likelihood function, the log-likelihood function is presented. We assume that each shock $e_{it}$, $i=1,...,d$, follows the skewed $t$-distribution introduced by \cite{Hansen:1994} with zero mean, unit variance, $\nu_i>2$ degrees of freedom, and skewness controlled by the parameter $\lambda_i\in (-1, 1)$. The skewed $t$-distribution can flexibly capture fat tails and skewness, and as the Gaussian distribution is obtained as a special case, with $\nu_i=\infty$ and $\lambda_i=0$, the plausibility of the identifying Assumption~\ref{as:shocks} can be (informally) assessed based on the estimates of these parameters.

The density function of the skewed $t$-distribution, $st(\cdot; \nu_i, \lambda_i)$, is given as \citep[][Equations~(10)-(13)]{Hansen:1994}:
\begin{equation}\label{eq:skew_t_dens}
st(e_{it}; \nu_i, \lambda_i) = b_ic_i\left(1 + \frac{1}{\nu_i - 2}\left(\frac{b_ie_{it} + a_i}{1 - \mathbbm{1}\lbrace e_{it}<-a_i/b_i  \rbrace\lambda_i + \mathbbm{1}\lbrace e_{it}\geq -a_i/b_i \rbrace\lambda_i} \right)^2 \right)^{-(\nu_i + 1)/2}
\end{equation}
where $\mathbbm{1} \lbrace e_{it}<-a_i/b_i \rbrace$ is an indicator function that takes the value one if $e_{it}<-a_i/b_i$ and zero otherwise, and $\mathbbm{1}\lbrace e_{it}\geq -a_i/b_i\rbrace$ is an indicator function that takes the value one if $e_{it}\geq -a_i/b_i$ and zero otherwise. The constants $a_i, b_i$, and $c_i$ are defined as $a_i = 4\lambda_i c_i\left(\frac{\nu_i - 2}{\nu_i - 1}\right)$, $b_i = (1 + 3\lambda_i^2 - a_i^2)^{1/2}$, and $c_i = \frac{\Gamma \left(\frac{\nu_i + 1}{2}\right)}{(\pi(\nu_i - 2))^{1/2}\Gamma\left(\frac{\nu_i}{2}\right)}$, where $\Gamma(\cdot)$ is the Gamma function.

Indexing the observed data as $y_{-p+1},...,y_0,y_1,...,y_T$, the conditional log-likelihood function, conditional on the initial values $\boldsymbol{y}_0=(y_0,...,y_{-p+1})$, is given as
\begin{equation}\label{eq:loglik}
L(\boldsymbol{\theta})=\sum_{t=1}^T\log\left(|\det(B_{y,t})|^{-1}\prod_{i=1}^d st(I_{d,i}'B_{y,t}^{-1}(y_t - \mu_{y,t});\nu_i,\lambda_i)\right),
\end{equation}
where $I_{d,i}$ is the $i$th column of the $d$-dimensional identity matrix, $\mu_{y,t}$ is the conditional mean defined in Section~\ref{sec:stvar} and $B_{y,t}$ is the impact matrix defined in Equation~(\ref{eq:Bt}). The parameters of the model are collected to the vector $\boldsymbol{\theta}=(\phi_{1},...,\phi_{M},\varphi_1,...,\varphi_M,\sigma,\alpha,\nu)$, where $\varphi_m=(\text{vec}(A_{m,1}),...,$ $\text{vec}(A_{m,p}))$, $m=1,...,M$, $\sigma=(\text{vec}(B_1),...,\text{vec}(B_M))$, $\alpha$ contains the parameters governing the transition weights, and $\nu=(\nu_1,...,\nu_d,\lambda_1,...,\lambda_d)$ contains the degrees-of-freedom and skewness parameters. With logistic weights~(\ref{eq:alpha_mt_logistic}), $\alpha=(c,\gamma)$, whereas with exogenous weights, $\alpha$ is omitted from the parameter vector.

The penalized log-likelihood function is then defined as
\begin{equation}\label{eq:penloglik}
PL(\boldsymbol{\theta}) = L(\boldsymbol{\theta}) - P(\boldsymbol{\theta}),
\end{equation}
where $L(\boldsymbol{\theta})$ is defined in~(\ref{eq:loglik}) and $P(\boldsymbol{\theta})\geq 0$ is a penalization term that penalizes the log-likelihood function from parameter values that are close to entering or are in an uninteresting region of the parameter space. To focus on avoiding unstable estimates, we define the penalization term as
\begin{equation}\label{eq:penterm}
P(\boldsymbol{\theta}) = \kappa Td \sum_{m=1}^M\sum_{i=1}^{dp} \text{max}\lbrace 0, |\rho(\boldsymbol{A}_m(\boldsymbol{\theta}))_i| - (1 - \eta) \rbrace^2,
\end{equation}
where $|\rho(\boldsymbol{A}_m(\boldsymbol{\theta}))_i|$ is the modulus of the $i$th eigenvalue of the companion form AR matrix of Regime~$m$, the tuning parameter $\eta\in (0, 1)$ determines how close to the boundary of the stability region the penalization starts, and $\kappa>0$ determines the strength of the penalization. 

Whenever the companion form AR matrix of a regime has eigenvalues greater than $1 - \eta$ in modulus, the penalization term~(\ref{eq:penterm}) is greater than zero, and it increases in the modulus of these eigenvalues. Our penalization term incorporates the multiplicative coefficient $Td$ in an attempt to standardize the strength of penalization with respect to the number of observations $T$ and variables $d$, as the magnitude of the log-likelihood typically increases with these quantities. In our empirical application (Section~\ref{sec:empapp}), we use the tuning parameter values $\eta = 0.05$ and $\kappa = 0.2$, thus, penalizing parameter values outside the stability region significantly, while maintaining flexibility near the boundary. 

\subsection{Three-step estimation procedure}\label{sec:three-step}
\cite{Lanne+Meitz+Saikkonen:2017} propose a three-step estimation procedure for their linear SVAR model identified by non-Gaussianity. In the first step, least squares (LS) estimation produces preliminary estimates of certain parameters, which serve as initial values in numerical optimization algorithms used in the subsequent steps to maximize the log-likelihood function. The nonlinear dynamics of our structural STVAR model, however, add substantial challenges to the estimation problem by inducing a large number of modes to the (penalized) log-likelihood function. To address these challenges, we propose a modified version of the three-step estimation procedure in which multimodality is explicitly taken into account. 

Our three-step estimation procedure proceeds with the following steps:
\begin{enumerate}
  \item Estimate the autoregressive and weight function parameters $\boldsymbol{\theta}_1\equiv (\phi_{1},...,\phi_{M},\varphi_1,...,\varphi_M,\alpha)$ by (penalized) nonlinear least squares (NLS) (see Appendix~\ref{sec:nls} for details). Denote this estimate of $\boldsymbol{\theta}_1$ as $\hat{\boldsymbol{\theta}}_1^{NLS}$.\label{step:LS}
  
  \item Estimate the error distribution parameters $\boldsymbol{\theta}_2\equiv (\sigma,\nu)$ by (penalized) ML with a genetic algorithm conditionally on the (penalized) NLS estimate $\hat{\boldsymbol{\theta}}_1^{NLS}$. Denote the obtained estimate of $\boldsymbol{\theta}_2$ as $\hat{\boldsymbol{\theta}}_2^{GA}$.\label{step:GA}
  
  \item Estimate the full parameter vector $\boldsymbol{\theta}=(\boldsymbol{\theta}_1,\boldsymbol{\theta}_2)$ by (penalized) ML by initializing a gradient based optimization algorithm 
   from $(\hat{\boldsymbol{\theta}}_1^{NLS},\hat{\boldsymbol{\theta}}_2^{GA})$.\label{step:VA}
\end{enumerate}

Step~\ref{step:LS} obtains initial estimates for the AR and weight function parameters, and it is comparable to Step~1 of \cite{Lanne+Meitz+Saikkonen:2017} but with the nonlinearity of the model explicitly accounted for. Notably, the NLS estimation often produces estimates that do not satisfy the usual stability condition for each of the regimes, and thus allowing for instability (but penalizing it) facilitates utilization of the NLS estimates in our three-step procedure. Step~\ref{step:GA} is comparable to Step~2 of \cite{Lanne+Meitz+Saikkonen:2017}, but the multimodality of the (penalized) log-likelihood function is taken into account by making use of a robust estimation algorithm that is able to escape from local maxima. 
Step~\ref{step:VA} finalizes the estimation by initializing a standard gradient based optimization algorithm from the initial estimates obtained from the previous steps. 

Due to the high multimodality of the (penalized) log-likelihood function, we recommend running Steps~\ref{step:GA} and~\ref{step:VA} a large number of times to improve the reliability of the results. Since a large number of estimation rounds are run, the estimation procedure produces a set of estimates, some of which presumably corresponding to different local maximums of the (penalized) log-likelihood function. 
If there are multiple solutions with (penalized) log-likelihoods close to the greatest found (penalized) log-likelihood, the blended identification strategy discussed in Section~\ref{sec:labellingshocks} can be employed to address weak identification. 

To assess the performance of the penalized maximum likelihood (PML) estimator as well as our three-step estimation procedure, we conduct a small-scale Monte Carlo study, which is discussed in detail in Appendix~\ref{sec:montecarlo}. Since estimation is computationally demanding, we consider the simple bivariate (structural) LSTVAR model with $p=1$ and $M=2$. To evaluate how the PML estimator performs when the AR matrices are in the penalization region, we consider two specifications of both of LSTVAR models: one where the AR matrices are well inside the stability region and one where they lie clearly in the region where the penalization term is strictly positive. 
According to the results (see Table~\ref{tab:monteresults_lstvar} in Appendix~\ref{sec:montecarlo} for details), the estimation accuracy is slightly better when the AR matrices are inside the stability region. Both specifications exhibit some bias in small samples. However, since both the bias and the standard deviations of the estimates diminish as the sample size increases, our results align with (possible) consistency of the PML estimator.

\section{Empirical application}\label{sec:empapp}
Our empirical application studies the macroeconomic effects of climate policy uncertainty (CPU). In the related literature, \cite{Fried+Novan+Peterman:2022} develop a dynamic general equilibrium model incorporating beliefs about future climate policy, and find that an increased climate policy transition risk shifts investments toward cleaner capital while reducing total investments and output. \cite{Khalil+Strobel:2023} reach similar conclusions in their dynamic stochastic general equilibrium (DSGE) model, attributing the effects to financial institutions' aversion to uncertainty. \cite{Huang+Punzi:2024} also use a DSGE model and find that CPU shocks cause firms to delay investments, which reduces output and raises inflation. These findings are supported by their recursively identified Bayesian SVAR model.

We are interested in studying how the effects of the CPU shock vary depending on the level of general economic policy uncertainty (EPU), as the level of uncertainty may affect the behavior of economic agents. Such variation can be accommodated in our structural STVAR model by specifying a logistic transition weight function with the (lagged) level of EPU as the switching variable. This approach allows us to capture potential state-dependent effects, providing more detailed results 
than linear SVAR analysis.


\subsection{Data}

We consider a monthly U.S. dataset consisting of five variables and covering the time period from 1987:4 to 2024:12. Following \cite{Khalil+Strobel:2023} and \cite{Huang+Punzi:2024}, we measure CPU with the climate policy uncertainty index (CPUI) of \cite{Gavriilidis:2021}, which is constructed based on the amount of newspaper coverage on climate policy uncertainty topics. 
In order to control for general economic policy uncertainty, we also include the economic policy uncertainty index (EPUI) of \cite{Baker+Bloom+Davis:2016}, which is based on newspaper coverage on topics related to economic policy uncertainty as well as on components quantifying the present value of future scheduled tax code expirations and disagreement among professional forecasters over future government purchases and consumer prices. 

As a measure of real economic activity, we include the log of industrial production index (IPI), which is detrended by taking first differences. 
For measuring the price level, we use the log of consumer price index (CPI), likewise detrended by taking first differences. 
Finally, the monetary policy stance is measured with the effective Federal funds rate that is replaced by the \cite{Wu+Xia:2016} shadow rate for the zero-lower-bound periods (RATE). The series of the included variables are presented in Figure~\ref{fig:seriesplot}.\footnote{The CPUI and EPUI data are retrieved from \url{https://www.policyuncertainty.com}, whereas IPI, CPI, and the Federal funds rate are retrieved from the Federal Reserve Bank of St. Louis database and the \cite{Wu+Xia:2016} shadow rate from the Federal Reserve Bank of Atlanta's website. The availability of the CPUI data determines the beginning of our sample period.}

\begin{figure}[t]
    \centerline{\includegraphics[width=\textwidth - 2cm]{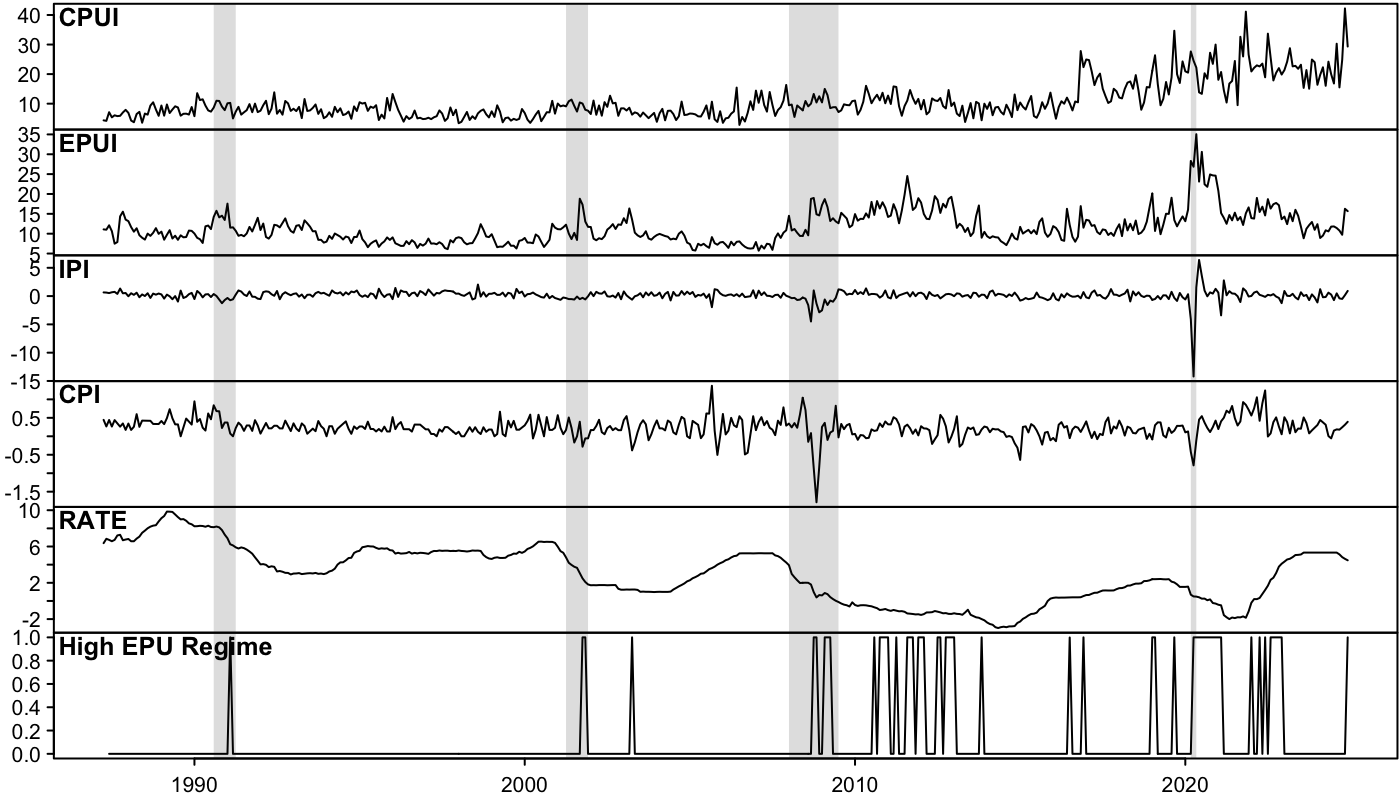}}
    \caption{Monthly U.S. time series covering the period from 1987:4 to 2024:12. From top to bottom, the variables are the climate policy uncertainty index (divided by $10$); economic policy uncertainty index (divided by $10$); the log-differences of industrial production index and consumer price index (multiplied by $100$); and the effective federal funds rate replaced by the \cite{Wu+Xia:2016} shadow rate for the zero lower bound periods. The bottom panel shows the transition weights of the High EPU Regime (Regime~2) of the fitted two-regime second-order structural LSTVAR model. The shaded areas indicate the U.S. recessions defined by the NBER.}
\label{fig:seriesplot}
\end{figure}

\subsection{The structural LSTVAR model}\label{sec:LSTVAR}
The structural STVAR model is specified from~(\ref{eq:stvarstruct})-(\ref{eq:Bt}) by defining the transition weights and distributions of the shocks. We employ the logistic transition weights with two regimes given in Equation~(\ref{eq:alpha_mt_logistic}). They facilitate smooth transitions between the regimes and enable analyzing how the effects of the CPU shock vary depending on the level of EPU, whose first lag is specified as the switching variable. 
%
%
%
We estimate the structural LSTVAR model by PML with the three-step procedure proposed in Section~\ref{sec:estimation}, and the shocks are assumed to follow skewed $t$-distributions. 
The autoregressive order $p=2$ is selected based on the Hannan-Quinn information criterion. The relatively low autoregressive order avoids overfitting, particularly in the regime accommodating fewer observations.

To deal with the problem of weak identification discussed in Section~\ref{sec:ident_nongaus} and~\ref{sec:labellingshocks}, we implement a blended identification strategy that imposes the following overidentifying restrictions on the impact matrices $B_1$ and $B_2$. First, for each variable, one of the shocks has a greater impact effect on that variable in both regimes than the other shocks, this shock is unique to that variable, and the impact effect has the same sign in both regimes. Second, the shock that has the greatest impact effect on output moves inflation in the same direction in both regimes (consistent with a demand shock). Third, the shock that has the greatest impact effect on inflation moves output in the opposite direction in both regimes (consistent with a supply shock). Fourth, the shock that has the greatest impact effect on interest rate moves output in the opposite direction in both regimes (consistent with a monetary policy shock). The shock that has the greatest impact effect on CPUI is deemed as the CPU shock, and to support our identification, we check that the recovered CPU shock aligns with several major historical events affecting CPU. In particular, we find that President Trump's election is accompanied by large positive CPU shocks in November 2016 and 2024. 
Moreover, there is a large negative CPU shock in January 2021, 
when President Biden signed executive orders to re-enter the Paris Agreement and launch a broad federal climate strategy.

The location parameter estimate $\hat{c}=16.16$, roughly implying that Regime~1 dominates when the level of EPU is not high, while Regime~2 prevails in times of high EPU. Hence, we label Regime~1 as the Low EPU Regime and Regime~2 as the High EPU Regime. The scale parameter estimate $\hat{\gamma}=64.23$, indicating that transitions between the regimes are quite fast, which can be observed also from the bottom panel of Figure~\ref{fig:seriesplot}, where the evolution of the fitted transition weight of the High EPU Regime is depicted.

The estimated vectors of the degrees-of-freedom and skewness parameters are $(2.77, 2.71, 3.00, 3.15, 2.03)$ and $(0.38, 0.41, -0.16, 0.15, 0.11)$, respectively, in the order of the variables CPUI, EPUI, IPI, CPI, and RATE. Clearly, none of the shocks is close to Gaussian, suggesting that Assumption~\ref{as:shocks} is plausible. Then, we check whether our model satisfies the stationarity condition by computing an upper bound for the joint spectral radius of the companion form AR matrices of the regimes (see Theorem~\ref{thm:stat} in Appendix~\ref{sec:stat}). 
The obtained upper bound is strictly less than one ($0.98$), implying that our model 
is ergodic stationary.\footnote{In addition, we checked that the estimated $B_1^{-1}B_2$ does not have real negative eigenvalues, so the additional condition mentioned in Footnote~\ref{footnote:B1B2cond} for logistic weights is satisfied.} 
Finally, based on graphical residual diagnostics, the overall adequacy of our model 
seems reasonable (for details, see Appendix~\ref{sec:adequacy}).

\subsection{Impulse response analysis}\label{sec:girfs}
The effects of the shocks in our structural LSTVAR model may depend on the initial state of the economy as well as on the sign and size of the shock, making the conventional way of calculating impulse responses unsuitable. Therefore, we consider the generalized impulse response function (GIRF) \citep{Koop+Pesaran+Potter:1996} that accommodates such features, defined as 
\begin{equation}\label{eq:girf}
\text{GIRF}(h,\delta_i,\mathcal{F}_{t-1}) = \text{E}[y_{t+h}|\delta_i,\mathcal{F}_{t-1}] - \text{E}[y_{t+h}|\mathcal{F}_{t-1}],
\end{equation}
where $h$ is the horizon. The first term on the right side of (\ref{eq:girf}) is the expected realization of the process at time $t+h$ conditionally on a structural shock of sign and size $\delta_i \in\mathbb{R}$ in the $i$th element of $e_t$ at time $t$ and the previous observations. The latter term on the right side is the expected realization of the process conditionally on the previous observations only. The GIRF thus expresses the expected difference in the future outcomes when the $i$th structural shock of sign and size $\delta_i$ arrives at time $t$ as opposed to all shocks being random. Since our model has the $p$-step Markov property, the conditioning set $\mathcal{F}_{t-1}$ can be replaced by the lag vector $\boldsymbol{y}_{t-1}=(y_{t-1},...,y_{t-p})$. The GIRF also facilitates tracing out the effects of the shocks on the transition weights $\alpha_{m,t}$, $m=1,...,M$, by replacing $y_{t+h}$ with $\alpha_{m,t+h}$ on the right side of Equation~(\ref{eq:girf}). 

To study the state-dependent effects of the CPU shock, we follow a procedure similar to \cite{Lanne+Virolainen:2025} and calculate the GIRFs conditional on a given regime dominating when the shock arrives. Specifically, for each Regime~$m\in \lbrace 1, 2\rbrace$, we take all the length $p$ histories $\boldsymbol{y}_{t-1}$ from the data for which the corresponding transition weight $\alpha_{m,t}$ is greater than $0.75$, indicating that Regime~$m$ is clearly dominant. 
For each such history, we take the corresponding CPU shock recovered from the data, and compute the GIRF using the Monte Carlo algorithm described in \cite{Lanne+Virolainen:2025}, Appendix~B. Finally, GIRFs corresponding to shocks of different sign and size are made comparable by scaling them to correspond to a $5$ point instantaneous increase of CPUI.

\begin{figure}[!t]
    \centerline{\includegraphics[width=\textwidth - 2cm]{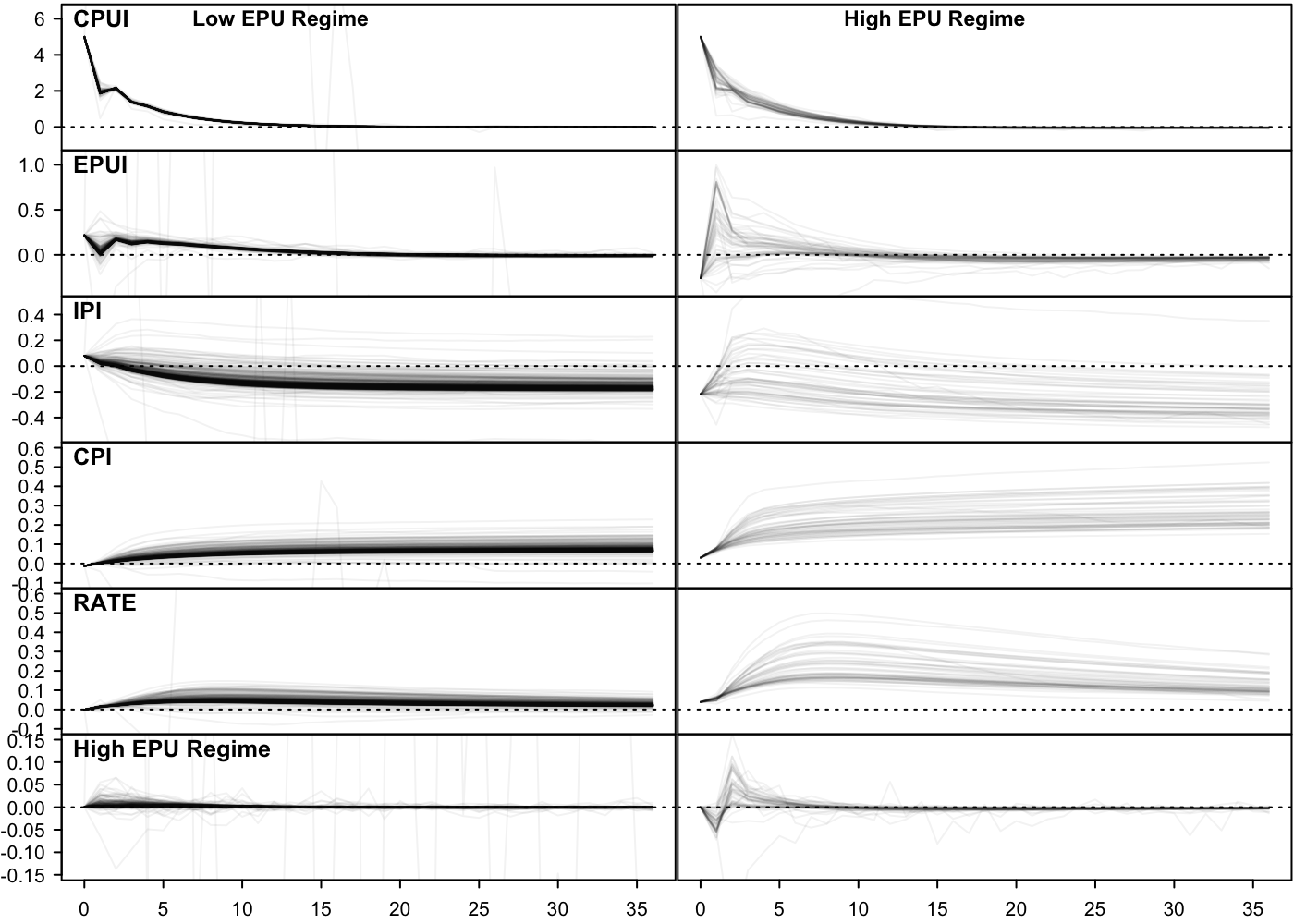}}
    \caption{Generalized impulse response functions to the identified CPU shock $h=0,1,....,36$ months ahead. From top to bottom, the responses of CPUI, EPUI, IPI, CPI, RATE, and the transition weights of the High EPU Regime are depicted in each row. The responses of IPI and CPI are accumulated to ($100\times$)log-levels. The left column shows the responses conditionally on the Low EPU Regime dominating when the shock arrives, whereas the right column shows the responses conditionally on the High EPU Regime dominating. All the GIRFs have been scaled so that the instantaneous increase in CPUI is $5$ points. Each GIRF is computed with a Monte Carlo algorithm similar to the one described in \citet[][Appendix~B]{Lanne+Virolainen:2025}.}
\label{fig:girfplot}
\end{figure}

Figure~\ref{fig:girfplot} illustrates the distribution of the GIRFs in each regime in a so-called "shotgun plot", which depicts each GIRF using a level of opacity such that the darkness of a region displays the concentration of GIRFs in it \citep[cf.][]{Lanne+Virolainen:2025}. 
The response of CPUI to a positive CPU shock follows a similar pattern in both regimes in the vast majority of the GIRFs. EPUI increases in both regimes, but response is stronger in the High EPU Regime. 


Consistent with the previous literature (\citealp{Fried+Novan+Peterman:2022}, \citealp{Khalil+Strobel:2023}, and \citealp{Huang+Punzi:2024}), a positive CPU shock decreases production in both regimes. \cite{Huang+Punzi:2024} attribute the decline in output to firms postponing investment, whereas \cite{Fried+Novan+Peterman:2022} and \cite{Khalil+Strobel:2023} emphasize reallocation of capital away from carbon-intensive sectors. 
We find that the contraction is slightly stronger in the High EPU Regime than in the Low EPU Regime, possibly reflecting greater risk aversion among firms and financial intermediaries during periods of elevated economic policy uncertainty.


In both regimes, a positive CPU shock increases prices, in line with \cite{Huang+Punzi:2024} and \cite{Khalil+Strobel:2023}. \cite{Huang+Punzi:2024} find that a positive CPU shock increases inflation by acting as a negative supply shock, whereas according to \cite{Khalil+Strobel:2023}, a positive CPU shock may decrease aggregate demand due to risk-averse households raising savings, while firms raise their prices preemptively to lower the probability of being stuck with a negative markup. However, we find that the inflationary effects are stronger in the High EPU Regime than in the Low EPU Regime. One possible explanation is that the firms might be more risk averse in the periods of high EPU and thereby react more sensitively. 

The response of the interest rate variable seems consistent with the inflationary effects of the CPU shock: it increases in both regimes, but clearly more so in the High EPU Regime. Finally, consistent with the GIRFs of EPUI, a positive CPU shock increases the transition weights of the High EPU Regime, more so in the High EPU Regime, where there is first a short-term decline in the weights. In other words, a positive CPU shock generally drives the economy towards the High EPU Regime. Overall, our results suggest that a positive CPU shock reduces production and raises inflation in both regimes, but its effects, particularly on inflation, are stronger under high economic policy uncertainty. 







\section{Conclusion}\label{sec:conclusions}
Linear structural vector autoregressive models can be identified statistically by non-Gaussianity, provided that the shocks are mutually independent and at most one of them is Gaussian \citep{Lanne+Meitz+Saikkonen:2017}. This paper is the first to extend identification by non-Gaussianity to nonlinear SVAR models. Specifically, we have shown that structural smooth transition vector autoregressive models, featuring a time-varying impact matrix defined as a weighted sum of regime-specific impact matrices, are likewise identified under these assumptions. While we have focused on exogenous, logistic, and threshold transition weights typically used in macroeconomic applications, our results can be extended to other suitable weight functions as well.



In addition to establishing the identification of the shocks, we have discussed the problem of labeling them. We have also employed blended identification \citep[cf.][]{Carriero+Marcellino+Tornese:2024} to address weak identification, which in our experience is often but not exclusively related to the ordering and signs of the columns of the regime-specific impact matrices. We have proposed estimating the model parameters by the method of penalized maximum likelihood, and we introduced a three-step estimation procedure by adapting the approach of \cite{Lanne+Meitz+Saikkonen:2017} to nonlinear SVAR models. Building on the results of \cite{Saikkonen:2008}, we have also provided a sufficient condition for ergodic stationarity of our structural STVAR model. The introduced methods are implemented in the accompanying R package sstvars \citep{sstvars}, which is available via the CRAN repository. 

In an empirical application to U.S. data from 1987:4 to 2024:12, we have studied the macroeconomic effects of the climate policy uncertainty shock using a two-regime logistic STVAR model. As a measure of climate policy uncertainty, we employ the CPU index \citep{Gavriilidis:2021} constructed based on the amount of newspaper coverage on climate policy uncertainty topics. To distinguish between periods of low and high economic policy uncertainty, we use the first lag of the EPU index \citep{Baker+Bloom+Davis:2016} as the switching variable. Our results are in line with the previous literature suggesting that a positive CPU shock reduces production and raises inflation (\citealp{Khalil+Strobel:2023} and \citealp{Huang+Punzi:2024}), while \cite{Fried+Novan+Peterman:2022} found that it decreases production. However, we found that while a positive CPU shock reduces production and raises inflation in times of both low and high EPU, its effects, particularly on inflation, are stronger during the periods of high EPU. 






\bibliography{/Users/savi/Documents/texrefs/masterrefs.bib}

\begin{thebibliography}{33}
\newcommand{\enquote}[1]{``#1''}
\providecommand{\natexlab}[1]{#1}
\providecommand{\url}[1]{\texttt{#1}}
\providecommand{\urlprefix}{URL }
\expandafter\ifx\csname urlstyle\endcsname\relax
  \providecommand{\doi}[1]{doi:\discretionary{}{}{}#1}\else
  \providecommand{\doi}{doi:\discretionary{}{}{}\begingroup
  \urlstyle{rm}\Url}\fi
\providecommand{\eprint}[2][]{\url{#2}}

\bibitem[{Anderson and Vahid(1998)}]{Anderson+Vahid:1998}
Anderson H., Vahid F. (1998).
\newblock \enquote{Testing multiple equation systems for common nonlinear
  components.}
\newblock \emph{Journal of Econometrics}, \textbf{84}(1), 1--36.

\bibitem[{Bacchiocchi and Fanelli(2015)}]{Bacchiocchi+Fanelli:2015}
Bacchiocchi E., Fanelli L. (2015).
\newblock \enquote{Identification in structural vector autoregressive models
  with structural changes, with an application to us monetary policy.}
\newblock \emph{Oxford Bulletin of Economics and Statistics}, \textbf{77}(6),
  761--779.

\bibitem[{Baker \emph{et~al.}(2016)Baker, Bloom, and
  Davis}]{Baker+Bloom+Davis:2016}
Baker S., Bloom N., Davis S. (2016).
\newblock \enquote{Measuring economic policy uncertainty.}
\newblock \emph{Quarterly Journal of Economics}, \textbf{131}(4), 1593--1636.

\bibitem[{Carriero \emph{et~al.}(2024)Carriero, Marcellino, and
  Tornese}]{Carriero+Marcellino+Tornese:2024}
Carriero A., Marcellino M., Tornese T. (2024).
\newblock \enquote{Blended identification in structural {VARs}.}
\newblock \emph{Journal of Monetary Economics}, \textbf{146}, {Article 103581}.

\bibitem[{Chang and Blondel(2013)}]{Chang+Blondel:2013}
Chang C.-T., Blondel V. (2013).
\newblock \enquote{An experimental study of approximation algorithms for the
  joint spectral radius.}
\newblock \emph{Numerical Algorithms}, \textbf{64}, 181--202.

\bibitem[{Cline and Phu(1998)}]{Cline+Phu:1998}
Cline D.~B., Phu H.-M.~H. (1998).
\newblock \enquote{Verifying irreducibility and continuity of a nonlinear time
  series.}
\newblock \emph{Statistics \& Probability Letters}, \textbf{40}(2), 139--148.

\bibitem[{Comon(1994)}]{Comon:1994}
Comon P. (1994).
\newblock \enquote{Independent component analysis, a new concept?}
\newblock \emph{Signal Processing}, \textbf{36}(3), 287--314.

\bibitem[{Fried \emph{et~al.}(2022)Fried, Novan, and
  Peterman}]{Fried+Novan+Peterman:2022}
Fried S., Novan K., Peterman W. (2022).
\newblock \enquote{Climate policy transition risk and the macroeconomy.}
\newblock \emph{European Economic Review}, \textbf{147}, Article 104174.

\bibitem[{Gavriilidis(2021)}]{Gavriilidis:2021}
Gavriilidis K. (2021).
\newblock \enquote{Measuring climate policy uncertainty. {Available} at {SSRN}:
  \url{https://ssrn.com/abstract=3847388}.}

\bibitem[{Gripenberg(1996)}]{Gripenberg:1996}
Gripenberg G. (1996).
\newblock \enquote{Computing the joint spectral radius.}
\newblock \emph{Linear algebra and its applications}, \textbf{234}, 43--60.

\bibitem[{Hansen(1994)}]{Hansen:1994}
Hansen B.~E. (1994).
\newblock \enquote{Autoregressive conditional density estimation.}
\newblock \emph{International Economic Review}, \textbf{35}(3), 705--730.

\bibitem[{Huang and Punzi(2024)}]{Huang+Punzi:2024}
Huang B., Punzi M. (2024).
\newblock \enquote{Macroeconomic impact of environmental policy uncertainty and
  monetary policy implications.}
\newblock \emph{Journal of Climate Finance}, \textbf{7}, Article 100040.

\bibitem[{Hubrich and Teräsvirta(2013)}]{Hubrich+Terasvirta:2013}
Hubrich K., Teräsvirta T. (2013).
\newblock \enquote{Thresholds and smooth transitions in vector autoregressive
  models.}
\newblock In T~Fomby, L~Killian, A~Murphy (eds.), \emph{VAR Models in
  Macroeconomics – New Developments and Applications: Essays in Honor of
  Christopher A. Sims (Advances in Econometrics)}, volume~32, pp. 273--326.
  Emerald Group Publishing Limited, Leeds.

\bibitem[{Jungers(2023)}]{Jungers:2023}
Jungers R. (2023).
\newblock \enquote{The {JSR} toolbox
  (\url{https://www.mathworks.com/matlabcentral/fileexchange/33202-the-jsr-toolbox}),
  matlab central file exchange.}

\bibitem[{Khalil and Strobel(2023)}]{Khalil+Strobel:2023}
Khalil M., Strobel F. (2023).
\newblock \enquote{Capital reallocation under climate policy uncertainty.}
\newblock \emph{Deutsche Bundesbank Discussion Paper No. 23/2023}.

\bibitem[{Kheifets and Saikkonen(2020)}]{Kheifets+Saikkonen:2020}
Kheifets I., Saikkonen P. (2020).
\newblock \enquote{Stationarity and ergodicity of vector {STAR} models.}
\newblock \emph{Econometric Reviews}, \textbf{39}(407--414), 1311--1324.

\bibitem[{Kilian and L{\"u}tkepohl(2017)}]{Kilian+Lutkepohl:2017}
Kilian L., L{\"u}tkepohl H. (2017).
\newblock \enquote{Structural vector autoregressive analysis.}
\newblock 1st edition. Cambridge University Press, Cambridge.

\bibitem[{Koop \emph{et~al.}(1996)Koop, Pesaran, and
  Potter}]{Koop+Pesaran+Potter:1996}
Koop G., Pesaran M., Potter S. (1996).
\newblock \enquote{Impulse response analysis in nonlinear multivariate models.}
\newblock \emph{Journal of Econometrics}, \textbf{74}(1), 119--147.

\bibitem[{Lanne \emph{et~al.}(2023)Lanne, Liu, and
  Luoto}]{Lanne+Liu+Luoto:2023}
Lanne M., Liu K., Luoto J. (2023).
\newblock \enquote{Identifying structural vector autoregression via leptokurtic
  economic shocks.}
\newblock \emph{Journal of Business \& Economic Statistics}, \textbf{41}(4),
  1341--1351.

\bibitem[{Lanne and Luoto(2021)}]{Lanne+Luoto:2021}
Lanne M., Luoto J. (2021).
\newblock \enquote{{GMM} estimation of non-{Gaussian} structural vector
  autoregression.}
\newblock \emph{Journal of Business \& Economic Statistics}, \textbf{39}(1),
  69--81.

\bibitem[{Lanne and L{\"u}tkepohl(2010)}]{Lanne+Lutkepohl:2010}
Lanne M., L{\"u}tkepohl H. (2010).
\newblock \enquote{Structural vector autoregressions with nonnormal residuals.}
\newblock \emph{Journal of Business \& Economic Statistics}, \textbf{28}(1),
  159--168.

\bibitem[{Lanne \emph{et~al.}(2017)Lanne, Meitz, and
  Saikkonen}]{Lanne+Meitz+Saikkonen:2017}
Lanne M., Meitz M., Saikkonen P. (2017).
\newblock \enquote{Identification and estimation of {non-Gaussian} structural
  vector autoregressions.}
\newblock \emph{Journal of Econometrics}, \textbf{196}(2), 288--304.

\bibitem[{Lanne and Virolainen(2025)}]{Lanne+Virolainen:2025}
Lanne M., Virolainen S. (2025).
\newblock \enquote{A {Gaussian} smooth transition vector autoregressive model:
  An application to the macroeconomic effects of severe weather shocks.}
\newblock \emph{Journal of Economic Dynamics and Control}, \textbf{178},
  105162.

\bibitem[{Lewis(2021)}]{Lewis:2021}
Lewis D. (2021).
\newblock \enquote{Identifying shocks via time-varying volatility.}
\newblock \emph{The Review of Economic Studies}, \textbf{88}(6), 3086--3124.

\bibitem[{L{\"u}tkepohl and Netšunajev(2017)}]{Lutkepohl+Netsunajev:2017}
L{\"u}tkepohl H., Netšunajev A. (2017).
\newblock \enquote{Structural vector autoregressions with smooth transitions in
  variances.}
\newblock \emph{Journal of Economic Dynamics \& Control}, \textbf{84}, 43--57.

\bibitem[{Meyn and Tweedie(1993)}]{Meyn+Tweedie:1993}
Meyn S., Tweedie R. (1993).
\newblock \enquote{Markov chains and stochastic stability.}
\newblock 1st edition. Springer-Verlag, London.

\bibitem[{Nielsen and Rahbek(2024)}]{Nielsen+Rahbek:2024}
Nielsen H., Rahbek A. (2024).
\newblock \enquote{Penalized quasi-likelihood estimation and model selection
  with parameters on the boundary of the parameter space.}
\newblock \emph{The Econometrics Journal}, \textbf{27}(1), 107--125.

\bibitem[{Rigobon(2003)}]{Rigobon:2003}
Rigobon R. (2003).
\newblock \enquote{Identification through heteroskedasticity.}
\newblock \emph{The Review of Economics and Statistics}, \textbf{85}(4),
  777--792.

\bibitem[{Saikkonen(2008)}]{Saikkonen:2008}
Saikkonen P. (2008).
\newblock \enquote{Stability of regime switching error correction models under
  linear cointegration.}
\newblock \emph{Econometric Theory}, \textbf{24}(1), 294--318.

\bibitem[{Tsay(1998)}]{Tsay:1998}
Tsay R. (1998).
\newblock \enquote{Testing and modeling multivariate threshold models.}
\newblock \emph{Journal of the American Statistical Association},
  \textbf{93}(443), 1188--1202.

\bibitem[{Virolainen(2025{\natexlab{a}})}]{sstvars}
Virolainen S. (2025{\natexlab{a}}).
\newblock \enquote{{sstvars}: Toolkit for reduced form and structural smooth
  transition vector autoregressive models.}
\newblock R package version 1.2.2 available at CRAN:
  \url{https://CRAN.R-project.org/package=sstvars}.

\bibitem[{Virolainen(2025{\natexlab{b}})}]{Virolainen:2025}
Virolainen S. (2025{\natexlab{b}}).
\newblock \enquote{A statistically identified structural vector autoregression
  with endogenously switching volatility regime.}
\newblock \emph{Journal of Business \& Economic Statistics}, \textbf{43}(1),
  44--54.

\bibitem[{Wu and Xia(2016)}]{Wu+Xia:2016}
Wu J., Xia F. (2016).
\newblock \enquote{Measuring the macroeconomic impact of monetary policy at the
  zero lower bound.}
\newblock \emph{Journal of Money, Credit and Banking}, \textbf{48}(2-3),
  253--291.

\end{thebibliography}

\pagebreak
\begin{appendices}
\renewcommand{\thefigure}{\thesection.\arabic{figure}}
\renewcommand{\thetable}{\thesection.\arabic{table}}
\setcounter{figure}{0}    
\setcounter{table}{0}

\newtheorem{appendixlemma}{Lemma}[section]
\newtheorem{appendixproposition}{Proposition}[section]
\newtheorem{appendixtheorem}{Theorem}[section]
\newtheorem{appendixassumption}{Assumption}[section]
\newtheorem{appendixcondition}{Condition}[section]
\newtheorem{appendixcorollary}{Corollary}[section]

\section{Identification of the shocks in structural TVAR models}\label{sec:identtvar}
The threshold VAR (TVAR) model of \cite{Tsay:1998} is obtained from~(\ref{eq:stvar1}) and (\ref{eq:stvar2}) by assuming that the transition weights are defined as 
\begin{equation}\label{eq:alpha_mt_threshold}
\alpha_{m,t} =
\left\lbrace\begin{matrix}
1 & \text{if} \ \ r_{m-1} < z_t \leq r_{m}, \\
0 & \text{otherwise}, \phantom{aaaaaaa}
\end{matrix}\right.
\end{equation}
where $r_m\in\mathbb{R}$, $m=1,...,M-1$, are the threshold parameters that satisfy $-\infty<r_1< ... <r_{M-1}<\infty$, $r_0\equiv-\infty$, $r_M\equiv\infty$, and $z_t\in\{y_{it-j}, i=1,...,d,j=1,...,p \}$ is the switching variable, which we assume the be a lagged endogenous variable (up to the lag $p$). At each $t$‚ the model defined in Equations~(\ref{eq:stvar1}), (\ref{eq:stvar2}) and (\ref{eq:alpha_mt_threshold}) reduces to a linear VAR corresponding to one of the regimes that is determined according to the level of the switching variable $z_t$. 
Notably, the TVAR model can be obtained as special case of the LSTVAR model when the smoothing parameter tends to infinity, i.e., $\gamma=\infty$ in~(\ref{eq:alpha_mt_logistic}). Also, exogenous transition weights may sometimes be binary or close-to-binary. Therefore, it is useful to establish the identification result also in such special cases of the STVAR model. 


Suppose the transition weights are binary, $\alpha_{m,t}\in\lbrace 0, 1\rbrace$, for all $t$ and $m=1,...,M$, specifically, either exogenous or of the threshold form~(\ref{eq:alpha_mt_threshold}). Then, at each $t$, the process is completely in one of the regimes, making the impact matrix $B_{y,t}=B_m$ for the active regime~$m$ with $\alpha_{m,t}=1$. If the allocation of the time periods to the regimes is uniquely identified, it then follows from Lemma~\ref{lemma:Bt_at_each_t} that $B_1,...,B_M$ are identified up to ordering and signs of their columns. 
Thus, to obtain identification up to ordering and signs, identification of the threshold and AR parameters is established in the following proposition, which is proven in Appendix~\ref{sec:proofpropotvar}.

\begin{appendixproposition}\label{prop:tvar_ident}
Consider the STVAR model defined in Equations~(\ref{eq:stvarstruct}) and (\ref{eq:Bt}) with $B_1,...,B_M$ invertible; Assumption~\ref{as:shocks} satisfied; and $\alpha_{m,t}$ either of the threshold form~(\ref{eq:alpha_mt_threshold}) or exogenous (nonrandom) with $\alpha_{m,t}\in\lbrace 0, 1\rbrace$ for all $t$ and $m=1,...,M$. Suppose the following conditions hold:
\begin{enumerate}[label=(\Alph*)]
\item $(\phi_{m},\text{vec}(A_{m,1}),...,\text{vec}(A_{m,p})) \neq (\phi_{n},\text{vec}(A_{n,1}),...,\text{vec}(A_{n,p}))$ for all $m\neq n \in \lbrace 1,...,M\rbrace$,\label{cond:tvar_different_arpars}
\item (for exogenous weights) there exists $M$ indices $t_1,...,t_M\in\lbrace 1,....,T\rbrace$ such that the corresponding coefficient vectors $(\alpha_{1,t_i},...,\alpha_{M,t_i})$, $i=1,...,M$, are linearly independent.\label{cond:tvar_linind}
\end{enumerate}
Then, the parameters $\phi_{m},A_{m,1},...A_{m,p}$, $m=1,...,M$, and (for threshold models) $r_1,...,r_{M-1}$ are uniquely identified. Moreover, $B_1,...,B_M$ are identified up to ordering and signs of their columns. 
\end{appendixproposition}
%
Condition~\ref{cond:tvar_different_arpars} of Proposition~\ref{prop:tvar_ident} implies distinguishable regimes, whereas Condition~\ref{cond:tvar_linind} guarantees (for exogenous weights) that each regime prevails in at least one time period. Under these conditions, the AR parameters and the partition of the time periods into the regimes are uniquely identified. Moreover, the regime-specific impact matrices $B_1,...,B_M$ are identified up to ordering and signs of their columns.\footnote{The identification is for every invertible $B_1,...,B_M$ and not just for almost every, as $B_{y,t}=B_m$ for some $m$ for all $t$, implying that $B_{y,t}$ is invertible for all invertible $B_m$. Thus, the result of Lemma~\ref{lemma:Bt_at_each_t} applies for all invertible $B_1,...,B_M$.}


By Proposition~\ref{prop:tvar_ident}, the shocks are identified by fixing the ordering and signs of the columns of $B_1,...,B_M$. In general, fixing them in one regime does not necessarily determine them in the others, and whether such constraints are purely normalizations or overidentifying depends on the distributions of the shocks. In practice, we recommend fixing the ordering and signs in all $B_1,...,B_M$ as a part of blended identification in Section~\ref{sec:labellingshocks}, which combines non-Gaussianity with supplementary information to address weak identification discussed in Sections~\ref{sec:ident_nongaus} and~\ref{sec:labellingshocks}.

%

\section{Stationarity and ergodicity}\label{sec:stat}
Establishing ergodic stationarity of the model is a common practice in time series econometrics, as it facilitates obtaining desirable asymptotic properties such as consistency of the estimator. \cite{Saikkonen:2008} derives a sufficient condition for ergodic stationarity of the smooth transition vector error correction (STVEC) model. \cite{Kheifets+Saikkonen:2020} make use of the results of \cite{Saikkonen:2008} to show that the stationarity condition readily applies for a STVAR model that can be obtained as a special case of the STVEC model of \cite{Saikkonen:2008}. However, in both \cite{Saikkonen:2008} and \cite{Kheifets+Saikkonen:2020}, the conditional covariance matrix of the process is defined as a weighted sum of positive definite error term covariance matrices of the regimes. As discussed in Section~\ref{sec:structstvarsub}, the conditional covariance matrix~(\ref{eq:condcovmat}) of our structural STVAR model has a different form, and therefore, our model is not nested to the model of \cite{Saikkonen:2008}. 
Nonetheless, we show in this section that the stationarity condition of \cite{Saikkonen:2008} applies to our model as well, given that the transition weights are either logistic or of the threshold form, with an extra assumption required for the logistic weights.

The sufficient stationarity condition is expressed in terms of the joint spectral radius (JSR) of certain matrices. The JSR of a finite set of square matrices $\mathcal{A}$ is defined by 
\begin{equation}
\rho(\mathcal{A}) = \underset{j\rightarrow \infty}{\limsup}\left(\underset{A\in \mathcal{A}^j}{\sup}\rho(A) \right)^{1/j},
\end{equation}
where $\mathcal{A}^j=\lbrace A_1A_2...A_j:A_i\in\mathcal{A}\rbrace$ and $\rho(A)$ is the spectral radius of the square matrix $A$.

Consider the companion form AR matrices of the regimes defined as
\begin{equation}\label{eq:boldA}
\boldsymbol{A}_m = 
\underset{(dp\times dp)}{\begin{bmatrix}
A_{m,1} & A_{m,2} & \cdots & A_{m,p-1} & A_{m,p} \\
I_d  & 0     & \cdots & 0            & 0 \\
0     & I_d  &             & 0            & 0 \\
\vdots &   & \ddots & \vdots    & \vdots \\
0     & 0     & \hdots & I_d         & 0
\end{bmatrix}}, \
m=1,...,M.
\end{equation}
The following assumption collects the sufficient conditions for establishing ergodicity, stationarity, and mixing properties of our structural STVAR model.
\begin{appendixassumption}\label{as:stat}
Suppose the following conditions hold:
\begin{enumerate}[label=(\Roman*)]
\item $\rho(\lbrace \boldsymbol{A}_1,...,\boldsymbol{A}_M \rbrace) < 1$,\label{cond:jsr}
\item The distribution of the $IID(0,I_d,\nu)$ random vector $e_t$ has a (Lebesgue) density that is bounded away from zero on compact subsets of $\mathbb{R}^d$.\label{cond:statshock}
\item The transition weights $\alpha_{m,t}$, $m=1,...,M$, follow either the threshold~\eqref{eq:alpha_mt_threshold} or logistic~\eqref{eq:alpha_mt_logistic} process with an endogenous switching variable $z_t\in\lbrace y_{it-j}, i=1,...,d, j=1,...,p \rbrace$, and\label{cond:statweights}
\item for logistic weights (with $M=2$ assumed), the matrix $B_1^{-1}B_2$ has no negative real eigenvalues.\label{cond:logB1B2}
\end{enumerate}
\end{appendixassumption}
Condition~\ref{cond:jsr} states that the JSR of the companion form AR matrices of the regimes is strictly less than one, and it is analogous to Condition~(19) of \cite{Saikkonen:2008}. Note that this condition is sufficient but not necessarily necessary for ergodic stationarity of the model. 
Condition~\ref{cond:statshock} is analogous Assumption~1 of \cite{Saikkonen:2008}, and it is innocuous in practice but rules out bounded error distributions, for example. 
Condition~\ref{cond:statweights}, in turn, restricts the form of the regime weights to either the standard threshold or logistic processes with endogenous switching, ensuring piecewise constancy or continuity of the weights. 
Finally, Condition~\ref{cond:logB1B2} rules out singular convex combinations of the regime-specific impact matrices in the two-regime logistic case, which guarantees positive definiteness of the conditional covariance matrix for all values of the transition weights. In the threshold case, positive definiteness is automatically guaranteed since the impact matrix coincides with one of the regime-specific matrices, each assumed invertible.

The following theorem (proven in Appendix~\ref{sec:proofthmstat}), which is analogous to Theorem~1 in \cite{Saikkonen:2008} and Theorem~1 in \cite{Kheifets+Saikkonen:2020}, states the results. 
\begin{appendixtheorem}\label{thm:stat}
Consider the STVAR process $y_t$ defined in Equations~(\ref{eq:stvarstruct}) and (\ref{eq:Bt}), and suppose Assumptions~\ref{as:shocks} and \ref{as:stat} are satisfied. Then, the process $\boldsymbol{y}_{t}=(y_t,...,y_{t-p+1})$ is a $(1 + ||x||^2)$-geometrically ergodic Markov chain. Thus, there exists a choice of initial values $y_{-p+1},...,y_0$ such that the process $y_t$ is strictly stationary, second-order stationary, and $\beta$-mixing with geometrically decaying mixing numbers. 
\end{appendixtheorem}

Assumption~\ref{as:stat} is sufficient for ergodic stationarity of our structural STVAR model, but Assumption~\ref{as:stat}\ref{cond:jsr} is computationally demanding to verify in practice \citep[see, e.g.,][]{Chang+Blondel:2013}. This makes it poorly suitable for restricting the parameter space in the numerical estimation of parameters discussed in Section~\ref{sec:estimation}. During estimation, this condition would have to be verified repeatedly a very large number of times, making the estimation extremely tedious. Hence, it is useful to make use of the following condition in the estimation, which is necessary for Assumption~\ref{as:stat}\ref{cond:jsr}.
\begin{appendixcondition}\label{cond:necessary}
$\max\lbrace \rho(\boldsymbol{A}_1),...,\rho(\boldsymbol{A}_M)\rbrace<1$,
\end{appendixcondition}
where $\rho(\boldsymbol{A}_m)$ is the spectral radius of $\boldsymbol{A}_m$, $m=1,...,M$. 

Condition~\ref{cond:necessary} states that the usual stability condition is satisfied by each of the regimes. It is necessary for Assumption~\ref{as:stat}\ref{cond:jsr}, as $\max(\rho(\boldsymbol{A}_1),...,\rho(\boldsymbol{A}_M))\leq \rho(\lbrace \boldsymbol{A}_1,...,\boldsymbol{A}_M \rbrace)$ \citep[see][and the references therein]{Kheifets+Saikkonen:2020}. After estimation, Assumption~\ref{as:stat}\ref{cond:jsr} can be checked for the solutions of interest.

Several methods have been proposed for computing bounds for the JSR, many of which are discussed in \cite{Chang+Blondel:2013}. The accompanying R package sstvars \citep{sstvars} implements the branch-and-bound method of \cite{Gripenberg:1996}, but the implementation is computationally very demanding if the matrices are large and a relatively tight bound is required. The JSR toolbox in MATLAB \citep{Jungers:2023}, in turn, automatically combines various methods to substantially enhance computational efficiency.

\section{Proofs}

\subsection{Proof of Lemma~\ref{lemma:invertibility}}\label{sec:proofinvertibility}
Before proving Lemma~\ref{lemma:invertibility}, we establish the following lemma that is made use of in the proof:
\begin{appendixlemma}\label{lemma:polyeq}
Suppose $P(x_1,...,x_k)$ is a finite-degree polynomial in the scalar variables $x_i\in\mathbb{R}$, $i=1,...,k$, that is not identically zero. Then, the set $\lbrace (x_1,...,x_k)\in\mathbb{R}^k | P(x_1,...,x_k) = 0 \rbrace$ has Lebesgue measure zero in $\mathbb{R}^k$.
\end{appendixlemma}
\begin{proof}
Proof by induction on $k$. For $k=1$, the claim clearly holds, as a nontrivial univariate polynomial has finitely many zeros. For the induction step, suppose the set $\lbrace (x_1,...,x_{k-1})\in\mathbb{R}^{k-1} | P(x_1,...,x_{k-1}) = 0 \rbrace$ has Lebesgue measure zero in $\mathbb{R}^{k-1}$. Write a point $x\in\mathbb{R}^k$ as $x=(x^*,z)$ where $x^*=(x_1,...,x_{k-1})\in\mathbb{R}^{k-1}$ and $z\in\mathbb{R}$ is the last component. Then, we can view $P(x^*,z)$ as a polynomial in $z$, whose coefficients are polynomials in $x^*$:
\begin{equation}
 P(x^*,z)=\sum_{i=0}^s a_i(x^*) z^i,
\end{equation}
where $s$ is the degree of $P$ in the variable $z$ and each coefficient $a_i(x^*)$ is itself a polynomial in $k-1$ variables. 

Define $\mathcal{E}\equiv\lbrace x^* \in\mathbb{R}^{k-1} : a_0(x^*) = a_1(x^*) = \cdots = a_s(x^*) = 0\rbrace$. If $x^*\in \mathcal{E}$, then all coefficients vanish, so $P(x^*,z)=0$ for all $z$. 
If $x^*\notin \mathcal{E}$, then $P(x^*,z)$ is a nontrivial univariate polynomial in $z$. Because $P$ is not identically zero, at least one coefficient $a_i(x^*)$, $i\in\{ 1,...,s\}$, is nontrivial. Then, by the induction hypothesis, the zero set of this nontrivial $a_i(x^*)$ (and hence $\mathcal{E}\subset \{x^*: a_i(x^*)=0\}$) has measure zero in $\mathbb{R}^{k-1}$. Also, for each fixed $x^*\in\mathbb{R}^{k-1}$, define the set $\mathcal{S}(x^*)\equiv\{z\in\mathbb{R}: P(x^*,z)=0\}$. If $x^*\notin \mathcal{E}$, then $P(x^*,z)$ is a nontrivial univariate polynomial, so $\mathcal{S}(x^*)$ is a finite set of points in $\mathbb{R}$ and thereby has measure zero. If $x^*\in \mathcal{E}$, then $\mathcal{S}(x^*)=\mathbb{R}$. 

Now, the zero set of $P$ is $\mathcal{Z}\equiv \{(x^*,z)\in\mathbb{R}^k : P(x^*,z)=0\}$. Denoting by $\lambda_k(\cdot)$ the $k$-dimensional Lebesgue measure, we have by Fubini's theorem:
\begin{equation}\label{eq:indfubini}
\lambda_k(\mathcal{Z}) = \int_{\mathbb{R}^{k-1}\times\mathbb{R}}\mathbbm{1}\{(x^*,z)\in\mathcal{Z} \} dx^* dz = \int_{\mathbb{R}^{k-1}}\left(\int_{\mathbb{R}}\mathbbm{1}\{(x^*,z)\in\mathcal{Z} \} dz\right) dx^*,
\end{equation}
where $\mathbbm{1}\{(x^*,z)\in\mathcal{Z} \}$ is an indicator function that takes the value one if $(x^*,z)\in\mathcal{Z}$ and the value zero otherwise. The inner integral on the right side of~\eqref{eq:indfubini} is the one-dimensional Lebesgue measure of the slice $\{z : (x^*,z)\in\mathcal{Z}\}=S(x^*)$, so $\lambda_k(\mathcal{Z}) = \int_{\mathbb{R}^{k-1}}\lambda_1(S(x^*)) dx^*$.
If $x^*\notin \mathcal{E}$, the slice $S(x^*)$ is finite, implying that $\lambda_1(S(x^*))=0$. On the other hand, if $x^*\in \mathcal{E}$, the slice equals $\mathbb{R}$, implying that $\lambda_1(S(x^*))=\infty$. However, since $\mathcal{E}$ has measure zero in $\mathbb{R}^{k-1}$, it follows that $\lambda_1(S(x^*))$ is zero almost everywhere, and hence, $\lambda_k(\mathcal{Z})=0$, concluding the proof. 
\end{proof}


To prove Lemma~\ref{lemma:invertibility}, we start by showing that for all $\alpha_{1,t},...,\alpha_{M,t}$, the matrix $B_{y,t}=\sum_{m=1}^M \alpha_{m,t}B_m$ is invertible almost everywhere in $[B_1:...:B_M] \in \mathbb{R}^{d\times dM}$. 
First, note that for all $\alpha_{1,t},...,\alpha_{M,t}$ there exists (invertible) matrices $B_1,...,B_M$ such that $B_{y,t}$ is invertible, e.g., $B_1=...=B_M=I_d$, so $B_{y,t}=I_d$ (as $\sum_{m=1}^M \alpha_{m,t}=1$ for all $t$). The set of points $[B_1:...:B_M] \in\mathbb{R}^{d\times dM}$ 
for which $B_{y,t}$ is singular is $S_t\equiv\lbrace ([B_1:...:B_M] \in\mathbb{R}^{d\times dM} | \det (\sum_{m=1}^M \alpha_{m,t}B_m)=0 \rbrace$. Since the determinant of a matrix is polynomial in the entries of that matrix, and the entries of $\sum_{m=1}^M \alpha_{m,t}B_m$ are linear combinations of the entries of $B_1,...,B_M$, it follows that the determinant of $\sum_{m=1}^M \alpha_{m,t}B_m$ is a polynomial in the entries of $B_1,...,B_M$. Denote this polynomial as $P_t(B_1,...,B_M)\equiv \det (\sum_{m=1}^M \alpha_{m,t}B_m)$. Note that whenever $\alpha_{m,t}=0$, the impact matrices $B_m$ of the corresponding regimes drop out from $P_t(B_1,...,B_M)$, but in this case $P_t(B_1,...,B_M)$ can still be interpreted as a polynomial in the entries of $B_1,...,B_M$ with the multipliers of the dropped entries being zero. It then follows from Lemma~\ref{lemma:polyeq} that $S_t$ has Lebesgue measure zero in $\mathbb{R}^{d\times dM}$.


Finally, note that since the above result holds for all $\alpha_{1,t},...,\alpha_{M,t}$, the set $S_t$ has Lebesgue measure zero in $\mathbb{R}^{d\times dM}$ for each $t$. Hence, the union $\cup_{t\in\mathbb{N}}S_t$ is a union of a countable number of measure zero sets and thus itself has Lebesgue measure zero. It follows that the matrix $B_{y,t}$ is invertible for all $t$ almost everywhere in $ [B_1:...:B_M] \in \mathbb{R}^{d\times dM}$, concluding the proof.\qed


\subsection{Proof of Lemma~\ref{lemma:Bt_at_each_t}}\label{sec:prooflemmaBteach}
The first part of our proof is largely similar to the proof of Proposition~1 in \cite{Lanne+Meitz+Saikkonen:2017} building on results from independent component analysis \citep{Comon:1994}. First, note that from Equation~(\ref{eq:stvarstruct}), we can rewrite the model as 
\begin{equation}\label{eq:stvarsimp}
y_t - \phi_{y,t} - \boldsymbol{A}_{y,t}\boldsymbol{y}_{t-1} = B_{y,t}e_t,
\end{equation}
where $\phi_{y,t} \equiv \sum_{m=1}^M \alpha_{m,t}\phi_{m}$ $(d\times 1)$; $\boldsymbol{A}_{y,t}\equiv[A_{y,t,1}:...:A_{y,t,p}]$ $(d\times dp)$; $A_{y,t,i}\equiv \sum_{m=1}^M\alpha_{m,t}A_{m,i}$, $i=1,...,p$, are the time-varying AR matrices; and $\boldsymbol{y}_{t-1}=(y_{t-1},...,y_{t-p})$. Then, suppose Equation~(\ref{eq:stvarsimp}) holds also for some other coefficients, transition weights, impact matrix, and shocks $\tilde{\phi}_{m},\tilde{A}_{m,1},...,\tilde{A}_{m,p},\tilde{\alpha}_{m,t}$, $m=1,...,M$, $\tilde{B}_{y,t}$, and $\tilde{e}_t$ (and  define $\tilde{\phi}_{y,t}$ and $\tilde{\boldsymbol{A}}_{y,t}$ analogously), i.e., $y_t - \tilde{\phi}_{y,t} - \tilde{\boldsymbol{A}}_{y,t}\boldsymbol{y}_{t-1} = \tilde{B}_{y,t}\tilde{e}_t$. By subtracting this identify from (\ref{eq:stvarsimp}), we obtain
\begin{equation}\label{eq:tildeproof1}
\tilde{\phi}_{y,t} - \phi_{y,t} + (\tilde{\boldsymbol{A}}_{y,t} - \boldsymbol{A}_{y,t})\boldsymbol{y}_{t-1} = B_{y,t}e_t - \tilde{B}_{y,t}\tilde{e}_t.
\end{equation}
Taking conditional expectation conditional on $\mathcal{F}_{t-1}$ gives
\begin{equation}\label{eq:tildeproof2}
\tilde{\phi}_{y,t} - \phi_{y,t} + (\tilde{\boldsymbol{A}}_{y,t} - \boldsymbol{A}_{y,t})\boldsymbol{y}_{t-1} = 0
\end{equation}
almost surely, as $E[B_{y,t}e_t|\mathcal{F}_{t-1}] = E[\tilde{B}_{y,t}\tilde{e}_t|\mathcal{F}_{t-1}]=0$ (by $e_t \ind \mathcal{F}_{t-1}$ and $E[e_t]=0$). By substituting this back to~(\ref{eq:tildeproof1}), we get 
\begin{equation}
B_{y,t}e_t = \tilde{B}_{y,t}\tilde{e}_t.
\end{equation}

Since $B_{y,t}$ is invertible almost everywhere in $[B_1:...:B_M] \in \mathbb{R}^{d\times dM}$ by Lemma~\ref{lemma:invertibility}, we can solve $e_t$ as $e_t = C\tilde{e}_t$, where $C=B_{y,t}^{-1}\tilde{B}_{y,t}$ is $\mathcal{F}_{t-1}$-measurable. By Assumption~\ref{as:shocks}, the random variables $\tilde{e}_{1t},...,\tilde{e}_{dt}$ are mutually independent and at most one of them is Gaussian. Also the random variables $e_{1t},...,e_{dt}$ are mutually independent and at most one of them is Gaussian. Therefore, conditionally on $\mathcal{F}_{t-1}$, by Darmois-Skitovich theorem \citep[Lemma~A.1 in][]{Lanne+Meitz+Saikkonen:2017}, at most one column of $C$ may contain more than one nonzero element. Suppose, say, the $k$th column of $C$ has at least two nonzero elements, $c_{ik}$ and $c_{jk}$, $i\neq j$. Then, $e_{it}=c_{ik}\tilde{e}_{kt} + \sum_{l=1,l\neq k}^{d}c_{il}\tilde{e}_{lt}$ and $e_{jt}=c_{jk}\tilde{e}_{kt} + \sum_{l=1,l\neq k}^{d}c_{jl}\tilde{e}_{lt}$ with $\tilde{e}_{kt}$ Gaussian (by Darmois-Skitovich theorem). This implies that $E[e_{it}e_{jt}]=c_{ik}c_{jk}\neq 0$ (since $E[\tilde{e}_{kt}^2]=1$), so $e_{it}$ and $e_{jt}$ are not independent, which is a contradiction. Hence, each column of $C$ has at most one nonzero element. 

Now, from the invertibility of $C$ it follows that each column of $C$ has exactly one nonzero element, and for the same reason, also that each row of $C$ has exactly one nonzero element. Therefore, there exists a permutation matrix $P$ and a diagonal matrix $D=\text{diag}(d_1,...,d_d)$ with nonzero diagonal elements such that $C=DP$. Together with $C=B_{y,t}^{-1}\tilde{B}_{y,t}$ this implies that $\tilde{B}_{y,t}=B_{y,t}DP$ and $\tilde{e}_t=P'D^{-1}e_t$, concluding that $B_{y,t}$ is unique up to ordering of its columns and multiplying each column by a constant. 

To complete the first part of the proof, it needs to be shown that the diagonal elements of $D$ are either $1$ or $-1$. By the normalization of the unit variances of $e_t$ and $\tilde{e}_t$, $\text{Cov}(\tilde{e}_t|\mathcal{F}_{t-1})=\text{Cov}(Ce_t|\mathcal{F}_{t-1})=I_d$, which is equivalent to $C\text{Cov}(e_t|\mathcal{F}_{t-1})C'=I_d$ and thereby $CC'=I_d$. By substituting $C=DP$ to this expression, we obtain $DPP'D'=I_d$ and by the orthogonality of permutation matrix $DD'=I_d$. That is, $d_{ii}^2=1$ for all $i=1,...,d$, so $d_{ii}=\pm 1$. Since diagonal elements of $D$ are, hence, either $1$ or $-1$, it follows that $B_{y,t}$ is unique up to ordering of its columns and changing all signs in a column almost everywhere in $[B_1:...:B_M] \in \mathbb{R}^{d\times dM}$. 


\textbf{Logistic weights}. Next, we show the "moreover" part for the logistic transition weights defined in~(\ref{eq:alpha_mt_logistic}) with $M=2$ and endogenous $z_t$, in particular, a lagged endogenous variable obtained as $z_t = a'\boldsymbol{y}_{t-1}$ with a fixed $a\neq 0$ (where $a$ has exactly one element equal to one and others equal to zero). Consider Equation~(\ref{eq:tildeproof2}) with two sets of parameters $\phi_{m},A_{m,1},...,A_{m,p}$, $m=1,...,M$, $c,\gamma$ and $\tilde{\phi}_{m},\tilde{A}_{m,1},...,\tilde{A}_{m,p}$, $m=1,...,M$, $\tilde{c},\tilde{\gamma}$. The proof proceeds in the following steps: (1) we show that the AR parameters $\phi_{m},A_{m,1},...,A_{m,p}$, $m=1,...,M$ are identified (without relying on possible identification of the transition weights), and then, (2) given identification of the AR parameters, the identification of the transition weights and the parameters $c,\gamma$ is established. 


\textbf{Part 1.} Denote $\boldsymbol{A}_m \equiv [A_{m,1}:...:A_{m,p}]$ $(d\times dp)$ and $\tilde{\boldsymbol{A}}_m \equiv [\tilde{A}_{m,1}:...:\tilde{A}_{m,p}]$ $(d\times dp)$, $m=1,2$. By making use of the identity $\alpha_{1,t} = 1 - \alpha_{2,t}$, write~(\ref{eq:tildeproof2}) as
\begin{align}
\begin{aligned}
\alpha_{1,t}\big[(\boldsymbol{A}_1-\boldsymbol{A}_2)\boldsymbol{y}_{t-1}+(\phi_1-\phi_2)\big]
&-\tilde{\alpha}_{1,t}\big[(\tilde{\boldsymbol{A}}_1-\tilde{\boldsymbol{A}}_2)\boldsymbol{y}_{t-1}+(\tilde{\phi}_1-\tilde{\phi}_2)\big]  \\
&+ \big[(\boldsymbol{A}_2-\tilde{\boldsymbol{A}}_2)\boldsymbol{y}_{t-1}+(\phi_2-\tilde{\phi}_2)\big]=0.\label{eq:2reg}
\end{aligned}
\end{align}

Since by assumption the transition weights depend only on the scalar switching variable $z_t = a'\boldsymbol{y}_{t-1}$ (a single component of $\boldsymbol{y}_{t-1}$), we can analyze~\eqref{eq:2reg} on level sets of $z_t$ where the weights become constant. On each such set, Equation~\eqref{eq:2reg} reduces to an affine relation in $\boldsymbol{y}_{t-1}$ with fixed coefficients, which allows us to exploit linear-algebraic arguments to identify the AR matrices and intercepts. Because Equation~\eqref{eq:2reg} holds with probability one for $\boldsymbol{y}_{t-1}$, and Assumption~\ref{as:shocks} implies that the distribution of $\boldsymbol{y}_{t-1}$ has a density strictly positive almost everywhere in $\mathbb{R}^{dp}$, it follows that Equation~\eqref{eq:2reg} holds almost everywhere in $\boldsymbol{y}_{t-1}\in\mathbb{R}^{dp}$. Below, we drop the time index from $\boldsymbol{y}_{t-1}$ and write simply $\boldsymbol{y}\in\mathbb{R}^{dp}$ for the generic lag vector when considering Equation~\eqref{eq:2reg} as a functional equation holding for almost every $\boldsymbol{y}\in\mathbb{R}^{dp}$.

For $v\in\mathbb{R}$, let $H_v \equiv \lbrace\boldsymbol{y} \in \mathbb{R}^{dp} : a'\boldsymbol{y} = v\rbrace$ be the level set of $z_t$. Since \eqref{eq:2reg} holds almost everywhere in $\boldsymbol{y} \in \mathbb R^{dp}$, it follows that for almost every value of $v=a'\boldsymbol{y}$, \eqref{eq:2reg} holds for almost every $\boldsymbol y \in H_v$. 
%
On $H_v$, the weights $\alpha_{1,t}=\alpha_1(v)$ and $\tilde{\alpha}_{1,t}=\tilde{\alpha}_1(v)$ are constants, so \eqref{eq:2reg} becomes
\begin{align}
\begin{aligned}
\alpha_1(v)\big[(\boldsymbol{A}_1-\boldsymbol{A}_2)\boldsymbol{y}+(\phi_1-\phi_2)\big]
&-\tilde{\alpha}_1(v)\big[(\tilde{\boldsymbol{A}}_1-\tilde{\boldsymbol{A}}_2)\boldsymbol{y}+(\tilde{\phi}_1-\tilde{\phi}_2)\big] \\
&+ (\boldsymbol{A}_2-\tilde{\boldsymbol{A}}_2)\boldsymbol{y}+(\phi_2-\tilde{\phi}_2) = 0. \label{eq:Hv}
\end{aligned}
\end{align}
Since the left side of \eqref{eq:Hv} is affine in $\boldsymbol{y}$, if it vanishes almost everywhere in $H_v$, it must vanish for all $\boldsymbol{y}\in H_v$. Thus, for almost every $v\in\mathbb{R}$, Equation~\eqref{eq:Hv} holds for all $\boldsymbol{y}\in H_v$, i.e., for all values of the $dp-1$ free entries orthogonal to $a$. 

To make the dependence on $v$ explicit, define
\begin{align}
L(v)& \equiv\alpha_1(v)(\boldsymbol{A}_1-\boldsymbol{A}_2)
-\tilde\alpha_1(v)(\tilde{\boldsymbol A}_1-\tilde{\boldsymbol{A}}_2)
+(\boldsymbol A_2-\tilde{\boldsymbol A}_2),\\
b(v)& \equiv\alpha_1(v)(\phi_1-\phi_2)
-\tilde\alpha_1(v)(\tilde\phi_1-\tilde\phi_2)
+(\phi_2-\tilde\phi_2),
\end{align}
so that \eqref{eq:Hv} reads $L(v)\boldsymbol{y}+b(v)=0$ for all $\boldsymbol{y}\in H_v$. Every $\boldsymbol{y} \in H_v$ can be written as
\begin{equation}
\boldsymbol{y} = \boldsymbol{y}_v(v) + U\boldsymbol{y}_{\smallsetminus v},
\end{equation}
where $\boldsymbol{y}_{\smallsetminus v}\in\mathbb{R}^{dp-1}$, $U\in\mathbb{R}^{dp\times(dp-1)}$ has columns spanning $\mathcal{S}\equiv\{u:a'u=0\}$ (so $a'U=0$), and $\boldsymbol{y}_v(v)$ is any particular vector satisfying $a'\boldsymbol{y}_v(v)=v$, e.g., $\boldsymbol{y}_v(v)\equiv \frac{v}{\|a\|^2}a$. Substituting gives, for almost every $v$, 
\begin{equation}
L(v)U\boldsymbol{y}_{\smallsetminus v} + (L(v)\boldsymbol{y}_v(v) + b(v)) = 0
\end{equation}
for all $\boldsymbol{y}_{\smallsetminus v}$. Hence, for almost every $v$, $L(v)U=0$ and $L(v)\boldsymbol{y}_v(v) + b(v)=0$. Each of  entry of $L(v)U$ and $L(v)\boldsymbol{y}_v(v) + b(v)$ is a real-analytic function of $v$ (since $\alpha_1(v),\tilde{\alpha}_1(v)$ are logistic and hence real-analytic, and $\boldsymbol{y}_v(v)$ is affine in $v$). It follows that by the Identity Theorem, both $L(v)U=0$ and $L(v)\boldsymbol{y}_v(v) + b(v)=0$ hold for all $v\in\mathbb{R}$. Consequently, \eqref{eq:Hv} holds for all $v$ and $\boldsymbol{y}\in H_v$. 

Subtracting \eqref{eq:Hv} for two (distinct) points $\boldsymbol{y}^{(1)},\boldsymbol{y}^{(2)}\in H_v$ eliminates the intercepts, giving
\begin{equation}\label{eq:Lv}
L(v)u = 0 \quad\text{for all }u \equiv \boldsymbol{y}^{(1)}-\boldsymbol{y}^{(2)} \in \mathcal{S},
\end{equation}
where $\mathcal{S} = \lbrace u : a'u = 0\rbrace$ is the $(dp-1)$-dimensional subspace of vectors orthogonal to $a$, i.e., the set of displacements in $\mathbb{R}^{dp}$ that keep $a'\boldsymbol{y}$ (and hence $z_t$) fixed. Note that the difference $u$ does not depend on the level $v$ of switching variable $z_t$, since the related component of $u$ is always zero by construction.

Now, observe that if $M\in\mathbb{R}^{d\times dp}$ satisfies $Mu=0$ for all $u\in \mathcal{S}$, then there exists a unique $\beta\in\mathbb{R}^d$ such that\footnote{To see this, pick any $w$ with $a'w=1$ and decompose $\boldsymbol{y} = u + (a'\boldsymbol{y})w$ with $u\in \mathcal{S}$. Then $M\boldsymbol{y} = (a'\boldsymbol{y})Mw$, and setting $\beta \equiv Mw$ gives $M\boldsymbol{y} = \beta a'\boldsymbol{y}$ for all $\boldsymbol{y}$, i.e., $M=\beta a'$. To show uniqueness, suppose $M=\tilde{\beta} a'$ holds also for some $\tilde{\beta}\in\mathbb{R}^d$. Then, $(\beta - \tilde{\beta}) a'=0$, and multiplying this by $w$ from the right gives $\beta - \tilde{\beta} = Mw - \tilde{\beta} = 0$, showing $\tilde{\beta} = Mw$ and hence $\tilde{\beta} =\beta$.} 
\begin{equation}\label{eq:linfact}
M\boldsymbol{y} = \beta a'\boldsymbol{y} \ \ \text{for all} \ \ \boldsymbol{y}.
\end{equation}

For any $u\in \mathcal{S}$, 
Equation~\eqref{eq:Lv} can be expressed as
\begin{equation}\label{eq:lin-u-v}
\alpha_1(v)(\boldsymbol{A}_1-\boldsymbol{A}_2)u
-\tilde{\alpha}_1(v)(\tilde{\boldsymbol{A}}_1-\tilde{\boldsymbol{A}}_2)u
+(\boldsymbol{A}_2-\tilde{\boldsymbol{A}}_2)u
=0.
\end{equation} 
For each component $k=1,...,d$, define
\begin{equation}
u^{(1)}_k \equiv [(\boldsymbol{A}_1-\boldsymbol{A}_2)u]_k,\quad
u^{(2)}_k \equiv [(\tilde{\boldsymbol{A}}_1-\tilde{\boldsymbol{A}}_2)u]_k,\quad
u^{(0)}_k \equiv [(\boldsymbol{A}_2-\tilde{\boldsymbol{A}}_2)u]_k.
\end{equation}
Then, the $k$th component of (\ref{eq:lin-u-v}) is 
\begin{equation}\label{eq:comp-linuv}
\alpha_1(v)u^{(1)}_k - \tilde{\alpha}_1(v)u^{(2)}_k + u^{(0)}_k = 0 \quad \forall v\in\mathbb{R}.
\end{equation}
Let $r(v) \equiv (\alpha_1(v), -\tilde{\alpha}_1(v), 1) \in \mathbb{R}^3$ and $w_k \equiv (u^{(1)}_k, u^{(2)}_k, u^{(0)}_k) \in \mathbb{R}^3$, so that~(\ref{eq:comp-linuv}) can be expressed as $r(v)'w_k=0$, i.e., $w_k$ is orthogonal to the set $\{r(v): v\in\mathbb{R}\}$. 

If $\alpha_1(v)=\tilde{\alpha}_1(v)$ for all $v$, then $r(v) = (\alpha_1(v), -\alpha_1(v), 1)$, so the span of $r(v)$ is only $2$-dimensional. Then, evaluating~(\ref{eq:comp-linuv}) at two distinct values of $v$ yields linear conditions that force $w_k=0$. If $\alpha_1(v)=\tilde{\alpha}_1(v)$ does not hold identically, the set $\{r(v): v\in\mathbb{R}\}$ spans the whole $\mathbb{R}^3$.\footnote{If not, then there exists a vector $(c_1,c_2,c_3)\neq 0$ with $c_1\alpha_1(v)-c_2\tilde{\alpha}_1(v)+c_3 = 0$ for all $v$. Taking the limit $v\rightarrow \infty$ yields $c_3=0$, as $\underset{v\rightarrow \infty}{\lim}\alpha_{1}(v)=\underset{v\rightarrow \infty}{\lim}\tilde{\alpha}_{1}(v)=0$. Then $c_1\alpha_1(v) = c_2\tilde{\alpha}_1(v)$ for all $v$, and taking the limit $v\rightarrow -\infty$ yields $c_1=c_2$, as $\underset{v\rightarrow -\infty}{\lim}\alpha_{1}(v)=\underset{v\rightarrow -\infty}{\lim}\tilde{\alpha}_{1}(v)=1$. This forces $\alpha_1(v) = \tilde{\alpha}_1(v)$, which is a contradiction.}
Therefore the only vector $w_k$ orthogonal to all $r(v)$ is the zero vector, so $w_k=0$. 

Since $w_k=0$ for every $k=1,...,d$, we have
\begin{equation}
(\boldsymbol{A}_1-\boldsymbol{A}_2)u=0, \ \ (\tilde{\boldsymbol{A}}_1-\tilde{\boldsymbol{A}}_2)u=0, \ \ \text{and} \ \ (\boldsymbol{A}_2-\tilde{\boldsymbol{A}}_2)u=0
\end{equation}
for all $u\in\mathcal S$. 
%
Applying the result given in~(\ref{eq:linfact}) then shows
\begin{equation}\label{eq:adiffbeta}
(\boldsymbol{A}_1-\boldsymbol{A}_2)\boldsymbol{y} = \beta_1a'\boldsymbol{y}, \ \ 
(\tilde{\boldsymbol{A}}_1-\tilde{\boldsymbol{A}}_2)\boldsymbol{y} = \beta_2a'\boldsymbol{y}, \ \ \text{and} \ \
(\boldsymbol{A}_2-\tilde{\boldsymbol{A}}_2)\boldsymbol{y} = \beta_0a'\boldsymbol{y}
\end{equation}
for some $\beta_0,\beta_1,\beta_2\in\mathbb{R}^d$ and all $\boldsymbol{y}$. Substituting these identities into~\eqref{eq:Hv} and taking $\boldsymbol{y}\in H_v$ (so $a'\boldsymbol{y}=v$) gives
\begin{equation}\label{eq:realanalv1}
\alpha_1(v)(\beta_1 v+\phi_1-\phi_2)
 - \tilde{\alpha}_1(v)(\beta_2 v+\tilde{\phi}_1-\tilde{\phi}_2)
 + (\beta_0 v+\phi_2-\tilde{\phi}_2)
 = 0
\end{equation}
for all $v$. 

For $v\neq 0$, dividing~(\ref{eq:realanalv1}) by $v$ and taking the limit $v\rightarrow \infty$ yields $\beta_0=0$, as $\underset{v\rightarrow \infty}{\lim}\alpha_{1}(v)=\underset{v\rightarrow \infty}{\lim}\tilde{\alpha}_{1}(v)=0$. Plugging this into the last equation in~(\ref{eq:adiffbeta}) gives $(\boldsymbol{A}_2-\tilde{\boldsymbol{A}}_2)\boldsymbol{y} = 0$ for all $\boldsymbol{y}$, and thus, $\boldsymbol{A}_2=\tilde{\boldsymbol{A}}_2$. With $\beta_0=0$, dividing~(\ref{eq:realanalv1}) by $v$ and taking the limit $v\rightarrow -\infty$ yields $\beta_1=\beta_2$, as $\underset{v\rightarrow -\infty}{\lim}\alpha_{1}(v)=\underset{v\rightarrow -\infty}{\lim}\tilde{\alpha}_{1}(v)=1$. Plugging this and $\boldsymbol{A}_2=\tilde{\boldsymbol{A}}_2$ in the first two equations in~(\ref{eq:adiffbeta}) yields $(\boldsymbol{A}_1-\tilde{\boldsymbol{A}}_1)\boldsymbol{y} = 0$ for all $\boldsymbol{y}$, and hence, $\boldsymbol{A}_1=\tilde{\boldsymbol{A}}_1$. 

Also, $\alpha_1(v) v \rightarrow 0$ as $v\rightarrow\infty$, 
since for $v>0$ we have $0\leq v\alpha_1(v) \leq ve^{-\gamma (v - c)} = \frac{v}{e^{\gamma v}}e^{\gamma c} \rightarrow 0$ as $v\rightarrow\infty$ by L'Hôpital's rule. Consequently, with $\beta_0=0$, taking the limit $v\rightarrow\infty$ in (\ref{eq:realanalv1}) gives $\phi_2=\tilde{\phi}_2$. 
Since the denominator of $v(\alpha_1(v) - \tilde{\alpha}_1(v)) = \frac{v\exp\lbrace\tilde{\gamma}(v - \tilde{c})\rbrace - v\exp\lbrace\gamma(v - c)\rbrace}{(1 + \exp\lbrace\gamma(v - c)\rbrace)(1 + \exp\lbrace\tilde{\gamma}(v - \tilde{c})\rbrace)}$ tends to one as $v\rightarrow -\infty$, it follows that $\underset{v\rightarrow -\infty}{\lim}v(\alpha_1(v) - \tilde{\alpha}_1(v)) = e^{-\tilde{\gamma}\tilde{c}}\underset{v\rightarrow -\infty}{\lim} ve^{\tilde{\gamma}v} - e^{-\gamma c}\underset{v\rightarrow -\infty}{\lim} ve^{\gamma v} = 0$ by applying L'Hôpital's rule. Therefore, with $\beta_0=0$, $\beta_1=\beta_2$, and $\phi_2=\tilde{\phi}_2$, taking the limit $v\rightarrow -\infty$ in (\ref{eq:realanalv1}) gives $\phi_1=\tilde{\phi}_1$. Thus, the AR parameters $\phi_m,A_{m,1},...,A_{m,p}$, $m=1,2$, are uniquely identified.
\textbf{Part 2.} Given identification of $\phi_m$ and $A_{m,i}$, we can substitute $\phi_m=\tilde{\phi}_m$ and $A_{m,i} =\tilde{A}_{m,i}$, $m=1,2$, $i=1,...,p$, to Equation~(\ref{eq:Hv}) to obtain that, for every $v\in\mathbb{R}$,
\begin{equation}\label{eq:Hv3}
(\alpha_{1}(v) - \tilde{\alpha}_{1}(v))\big[\mu_1(\boldsymbol{y}) - \mu_2(\boldsymbol{y})  \big] = 0 \ \ \text{for all} \ \ \boldsymbol{y}\in H_v=\{\boldsymbol{y} \in \mathbb{R}^{dp} : a'\boldsymbol{y} = v\},
\end{equation}
where $\mu_m(\boldsymbol{y})\equiv \phi_m + \boldsymbol{A}_m\boldsymbol{y}$, $m=1,2$. Hence, for each fixed $v$, either $\alpha_{1}(v) = \tilde{\alpha}_{1}(v)$ or $\mu_1(\boldsymbol{y}) = \mu_2(\boldsymbol{y})$ for all $\boldsymbol{y}\in H_v$.
Considering the latter case, the equality $\mu_1(\boldsymbol{y}) = \mu_2(\boldsymbol{y})$ can be written as 
\begin{equation}\label{eq:aidenteq}
(\phi_1 - \phi_2) + (\boldsymbol{A}_1 - \boldsymbol{A}_2)\boldsymbol{y} = 0.
\end{equation}
Subtracting for two (distinct) points $\boldsymbol{y}^{(1)},\boldsymbol{y}^{(2)}\in H_v$ eliminates the intercepts and gives
\begin{equation}
(\boldsymbol{A}_1 - \boldsymbol{A}_2)u=0 \ \ \text{for all} \ \ \boldsymbol{u}=\boldsymbol{y}^{(1)}-\boldsymbol{y}^{(2)}\in \mathcal{S},
\end{equation}
where $\mathcal{S} = \{u : a'u = 0 \}$ as before. Then, applying~(\ref{eq:linfact}) with $M=\boldsymbol{A}_1 - \boldsymbol{A}_2$ yields
\begin{equation}\label{eq:aidenteq2}
(\boldsymbol{A}_1 - \boldsymbol{A}_2)\boldsymbol{y} = \beta_5a'\boldsymbol{y} 
\end{equation}
for some $\beta_5\in\mathbb{R}^d$ and all $\boldsymbol{y}$. Substituting this to~(\ref{eq:aidenteq}) and taking $\boldsymbol{y}\in H_v$ so that $a'\boldsymbol{y}=v$ implies $(\phi_1 - \phi_2) + \beta_5v = 0.$

Thus, if $\beta_5\neq 0$, the equality $\mu_1(\boldsymbol{y}) = \mu_2(\boldsymbol{y})$ can hold only at a single level $v_0$ that satisfies $(\phi_1 - \phi_2) + \beta_5v_0 = 0$. If $\beta_5=0$, it follows from~(\ref{eq:aidenteq2}) that $\boldsymbol{A}_1 = \boldsymbol{A}_2$ and thus from~(\ref{eq:aidenteq}) that also $\phi_1 = \phi_2$, contradicting Condition~\ref{cond:different_arpars} of Lemma~\ref{lemma:Bt_at_each_t}. Therefore, it must be that $\alpha_{1}(v) = \tilde{\alpha}_{1}(v)$ for all $v\neq v_0$. However, since $\alpha_{1}(v),\tilde{\alpha}_{1}(v)$ are real-analytic in $v$ and coincide on $\mathbb{R}\setminus\{v_0\}$ that has a positive Lebesgue measure, they must agree everywhere by the Identity Theorem, concluding that $\alpha_{1,t} = \tilde{\alpha}_{1,t}$ identically.

To show identification of the parameters $\gamma$ and $c$, observe that the identity $\alpha_{1}(v)=\tilde{\alpha}_{1}(v)$ can be expressed as $(\tilde{\gamma}-\gamma) v = \tilde{\gamma}\tilde{c} - \gamma c$. Since this holds for all $v$, it holds for two distinct values of $v$, implying that $\gamma=\tilde{\gamma}$ and $c=\tilde{c}$, concluding the identification proof logistic weights.

\textbf{Exogenous nonrandom weights}. Suppose the transition weights $\alpha_{1,t},...,\alpha_{M,t}$ are exogenous (nonrandom), i.e., they are given. Consider Equation~(\ref{eq:tildeproof2}) with the two sets of parameters $\phi_{m},A_{m,1},$ $...,A_{m,p}$, $m=1,...,M$, and $\tilde{\phi}_{m},\tilde{A}_{m,1},...,\tilde{A}_{m,p}$, $m=1,...,M$ (and $\alpha_{m,t}=\tilde{\alpha}_{m,t}$ for all $m=1,...,M$ and $t$). Because Equation~\eqref{eq:tildeproof2} holds with probability one for $\boldsymbol{y}_{t-1}$, and Assumption~\ref{as:shocks} implies that the distribution of $\boldsymbol{y}_{t-1}$ has a density strictly positive almost everywhere in $\mathbb{R}^{dp}$, it follows that Equation~\eqref{eq:tildeproof2} holds almost everywhere in $\boldsymbol{y}_{t-1}\in\mathbb{R}^{dp}$. Therefore, 
\begin{equation}\label{eq:exocondmeanequal}
\sum_{m=1}^M\alpha_{m,t}(\tilde{\phi}_m - \phi_m) + \sum_{m=1}^M\alpha_{m,t}(\tilde{\boldsymbol{A}}_m - \boldsymbol{A}_m)\boldsymbol{y}_{t-1}=0
\end{equation}
almost everywhere in $\boldsymbol{y}_{t-1}\in\mathbb{R}^{dp}$, where $\boldsymbol{A}_m = [A_{m,1}:...:A_{m,p}]$ $(d\times dp)$ and $\tilde{\boldsymbol{A}}_m = [\tilde{A}_{m,1}:...:\tilde{A}_{m,p}]$ $(d\times dp)$, $m=1,...,M$ as before. 

Because the left side of~\eqref{eq:exocondmeanequal} is affine in $\boldsymbol{y}_{t-1}$, if it vanishes almost everywhere, it must vanish for all $\boldsymbol{y}_{t-1}\in\mathbb{R}^{dp}$. 
Thus, for each $t$,
\begin{equation}\label{eq:exocondmeanequal2}
\sum_{m=1}^M\alpha_{m,t}(\tilde{\phi}_m - \phi_m) = 0 \ \ \text{and} \ \ \sum_{m=1}^M\alpha_{m,t}(\tilde{\boldsymbol{A}}_m - \boldsymbol{A}_m)=0.
\end{equation}
By Condition~\ref{cond:linind} of Lemma~\ref{lemma:Bt_at_each_t}, there exists $M$ indices $t_1,...,t_M$ for which the $(M\times M)$ matrix of weights
\begin{equation}
W\equiv \begin{bmatrix}
\alpha_{1,t_1} & \cdots & \alpha_{M,t_1} \\
\vdots         & \ddots & \vdots \\
\alpha_{1,t_M} & \cdots & \alpha_{M,t_M}
\end{bmatrix}
\end{equation}
is invertible. Stacking the equations in~(\ref{eq:exocondmeanequal2}) for $t\in\{t_1,...,t_M\}$ yields two full-rank systems of linear equations $(W\otimes I_{d^2p})(\text{vec}(\tilde{\boldsymbol{A}}_1-\boldsymbol{A}_1),...,\text{vec}(\tilde{\boldsymbol{A}}_M-\boldsymbol{A}_M))=0$ and $(W\otimes I_{d})(\tilde{\phi}_1-\phi_1,...,\tilde{\phi}_M-\phi_M)=0$. Hence, $\tilde{\boldsymbol{A}}_m=\boldsymbol{A}_m$ and $\tilde{\phi}_m=\phi_m$ for all $m=1,...,M$, showing that the AR parameters are uniquely identified and concluding the proof.\qed

\subsection{Proof of Proposition~\ref{prop:stvar_ident}}\label{sec:proofstvarident}

By Lemma~\ref{lemma:Bt_at_each_t}, the transition weights $\alpha_{1,t},...,\alpha_{M,t}$ are uniquely identified regardless of whether $B_1,...,B_M$ are identified. Hence, the transition weights can be treated as uniquely identified in this proof, which consists of two parts. \textbf{Part~1}. We show that if the ordering and signs of the columns of $B_1$ are fixed, then changing the ordering or signs of the columns of any of $B_2,...,B_M$ leads to an impact matrix $B_{y,t}$ that, at some $t$, cannot be obtained by reordering or changing the signs of the columns of the impact matrix associated with the original (i.e., unmodified) matrices $B_1,...,B_M$.
\textbf{Part~2}. We combine this result with Lemma~\ref{lemma:Bt_at_each_t} 
to show that the matrices $B_1,...,B_M$ are uniquely identified. 

\textbf{Part 1}. Denote by $\boldsymbol{\tau}_m=(\tau_{m1},...,\tau_{md})$ the permutation $\tau_{m1},...,\tau_{md}$ and $\boldsymbol{\iota}_m=(\iota_{m1},...,\iota_{md})$ the sign changes $\iota_{m1},...,\iota_{md}$ of the columns of $B_m$, respectively. Moreover, denote by $B_m(\boldsymbol{\tau}_m,\boldsymbol{\iota}_m)$ the invertible impact matrix $B_m$ corresponding to the permutation $\boldsymbol{\tau}_m$ and sign changes $\boldsymbol{\iota}_m$ of its columns. Then, consider the permutations $\boldsymbol{\tau}_1,\boldsymbol{\tau}_2,...,\boldsymbol{\tau}_M$ and $\boldsymbol{\tau}_1,\tilde{\boldsymbol{\tau}}_2,...,\tilde{\boldsymbol{\tau}}_M$ as well as the sign changes $\boldsymbol{\iota}_1,\boldsymbol{\iota}_2,...,\boldsymbol{\iota}_M$ and $\boldsymbol{\iota}_1,\tilde{\boldsymbol{\iota}}_2,...,\tilde{\boldsymbol{\iota}}_M$. We start by showing that 
$B_{y,t}(\boldsymbol{\tau},\boldsymbol{\iota})\equiv \alpha_{1,t}B_1(\boldsymbol{\tau}_1,\boldsymbol{\iota}_1) + \sum_{m=2}^M\alpha_{m,t}B_m(\boldsymbol{\tau}_m,\boldsymbol{\iota}_m)$ cannot be obtained by reordering or changing the signs of the columns of $B_{y,t}(\tilde{\boldsymbol{\tau}},\tilde{\boldsymbol{\iota}})\equiv\alpha_{1,t}B_1(\boldsymbol{\tau}_1,\boldsymbol{\iota}_1) + \sum_{m=2}^M\alpha_{m,t}B_m(\tilde{\boldsymbol{\tau}}_m,\tilde{\boldsymbol{\iota}}_m)$ almost everywhere in $[B_1,...,B_M]\in\mathbb{R}^{d\times dM}$, whenever $\boldsymbol{\tau}_m\neq \tilde{\boldsymbol{\tau}}_m$ or $\boldsymbol{\iota}_m\neq \tilde{\boldsymbol{\iota}}_m$ for any $m=2,...,M$.

Suppose $\boldsymbol{\tau}_m\neq \tilde{\boldsymbol{\tau}}_m$ or $\boldsymbol{\iota}_m\neq \tilde{\boldsymbol{\iota}}_m$ for any $m=2,...,M$. Since by invertibility none of the columns of $B_m$ can be equal or negatives of each other, this implies that $B_m(\boldsymbol{\tau}_m,\boldsymbol{\iota}_m) \neq B_m(\tilde{\boldsymbol{\tau}}_m,\tilde{\boldsymbol{\iota}}_m)$ for some $m=2,...,M$. At each time period $t$, the impact matrix $B_{y,t}(\tilde{\boldsymbol{\tau}},\tilde{\boldsymbol{\iota}})$ is obtained by reordering or changing the signs of the columns of $B_{y,t}(\boldsymbol{\tau},\boldsymbol{\iota})$ if and only if 
\begin{equation}
\alpha_{1,t}B_1(\boldsymbol{\tau}_1,\boldsymbol{\iota}_1) + \sum_{m=2}^M\alpha_{m,t}B_m(\boldsymbol{\tau}_m,\boldsymbol{\iota}_m) = (\alpha_{1,t}B_1(\boldsymbol{\tau}_1,\boldsymbol{\iota}_1) + \sum_{m=2}^M\alpha_{m,t}B_m(\tilde{\boldsymbol{\tau}}_m,\tilde{\boldsymbol{\iota}}_m))D_tP_t
\end{equation}
for some $(d\times d)$ diagonal matrix $D_t$ with $\pm 1$ diagonal elements and some $(d\times d)$ permutation matrix $P_t$ (that may depend on $t$). Rearranging the above equation gives
\begin{equation}\label{eq:stvar_indeq1}
\sum_{m=2}^M\alpha_{m,t}(B_m(\boldsymbol{\tau}_m,\boldsymbol{\iota}_m) - B_m(\tilde{\boldsymbol{\tau}}_m,\tilde{\boldsymbol{\iota}}_m)D_tP_t) = -\alpha_{1,t}B_1(\boldsymbol{\tau}_1,\boldsymbol{\iota}_1)(I_d - D_tP_t).
\end{equation}

In the logistic case with $\gamma\in (0,\infty)$, $\alpha_{2,t}=[1+\exp\{-\gamma(z_t-c)\}]^{-1}\in(0,1)$ for all $t$, and hence also $\alpha_{1,t}=1-\alpha_{2,t}\in(0,1)$ for all $t$. From this, and in the case of exogenous weights from Condition~\ref{cond:stvar_alpha} of Proposition~\ref{prop:stvar_ident}, it follows that the weight vector $(\alpha_{1,t},...,\alpha_{M,t})$ takes a value for some $t$, say, $t^*$, such that none of its entries is zero. If $D_{t^*}P_{t^*}=I_d$, Equation~(\ref{eq:stvar_indeq1}) implies that
\begin{equation}\label{eq:stvar_indeq2}
\sum_{m=2}^M\alpha_{m,t^*}(B_m(\boldsymbol{\tau}_m,\boldsymbol{\iota}_m) - B_m(\tilde{\boldsymbol{\tau}}_m,\tilde{\boldsymbol{\iota}}_m)) = 0.
\end{equation}
Because $B_m(\boldsymbol{\tau}_m,\boldsymbol{\iota}_m) \neq B_m(\tilde{\boldsymbol{\tau}}_m,\tilde{\boldsymbol{\iota}}_m)$ for some $m=2,...,M$, the set of matrices $B_2,...,B_M$ that satisfy Equation~(\ref{eq:stvar_indeq2}) is the solution space of a nontrivial homogeneous linear system in the entries of $B_2,...,B_M$. This solution space is a proper linear subspace of $\mathbb{R}^{d\times d(M-1)}$ and therefore has Lebesgue measure zero. Consequently, the set of $[B_1:...:B_M]$ satisfying~(\ref{eq:stvar_indeq2}) has Lebesgue measure zero in $\mathbb{R}^{d\times dM}$. 

Suppose then $D_{t^*}P_{t^*}\neq I_d$ and denote the left side of Equation~(\ref{eq:stvar_indeq1}) as $C_{t}\in\mathbb{R}^{d\times d}$. At time period $t^*$, the right side of Equation~(\ref{eq:stvar_indeq1}) is nonzero (as $B_1$ is invertible), thus, implying that also $C_{t^*}$ is not a zero matrix. Equation~(\ref{eq:stvar_indeq1}) is thereby nontrivial and can be expressed as 
\begin{equation}\label{eq:stvar_indeq3}
B_1(\boldsymbol{\tau}_1,\boldsymbol{\iota}_1)(I_d - D_{t^*}P_{t^*}) = \alpha_{1,t^*}^{-1}C_{t^*}
\end{equation}
Since $\text{rank}(I_d - D_{t^*}P_{t^*})\geq 1$, the set of matrices $B_1$ that satisfy Equation~(\ref{eq:stvar_indeq3}) defines either an empty set or an affine subspace in $\mathbb{R}^{d\times d}$ with the dimension of this subspace being less than $d^2$. It follows that the set of matrices $B_1$ that satisfy Equation~(\ref{eq:stvar_indeq3}) has Lebesgue measure zero in $\mathbb{R}^{d\times d}$, and therefore, the set of matrices $B_1,...,B_M$ that satisfy Equation~(\ref{eq:stvar_indeq3}) has Lebesgue measure zero in $[B_1:...:B_M]\in\mathbb{R}^{d\times dM}$. 

By the above discussion, for any single transformation $D_{t^*}P_{t^*}$ of $B_{y,t^*}$ that satisfy Equation~(\ref{eq:stvar_indeq1}) at the time period $t^*$ has Lebesgue measure zero in $[B_1:\cdots : B_M]\in\mathbb{R}^{d\times dM}$. Since there exists a finite number of such transformations and a finite union of measure zero sets has measure zero, it follows that the set of matrices $B_1,...,B_M$ that satisfy Equation~(\ref{eq:stvar_indeq1}) for $t^*$, and thereby for all $t$, has Lebesgue measure zero in $[B_1:\cdots :B_M]\in\mathbb{R}^{d\times dM}$. This concludes that, almost everywhere in $[B_1,...,B_M]\in\mathbb{R}^{d\times dM}$, the matrix $B_{y,t}(\boldsymbol{\tau},\boldsymbol{\iota})$ cannot be obtained at all $t$ by reordering or changing the signs of the columns of $B_{y,t}(\tilde{\boldsymbol{\tau}},\tilde{\boldsymbol{\iota}})$ whenever $\boldsymbol{\tau}_m\neq \tilde{\boldsymbol{\tau}}_m$ or $\boldsymbol{\iota}_m\neq \tilde{\boldsymbol{\iota}}_m$ for any $m=2,...,M$. 

\textbf{Part 2}. Next, we make use of the above result together with Lemma~\ref{lemma:Bt_at_each_t} to show that $B_1,...,B_M$ are uniquely identified almost everywhere in $[B_1:\cdots :B_M]\in\mathbb{R}^{d\times dM}$. Consider the impact matrix $B_{y,t}=\sum_{m=1}^M \alpha_{m,t}B_m$ with some invertible $(d\times d)$ matrices $B_1,...,B_M$. By Lemma~\ref{lemma:Bt_at_each_t}, $B_{y,t}$ is uniquely identified at each $t$ up to ordering and signs of its columns (conditionally on $\mathcal{F}_{t-1}$). Since the ordering and signs of the columns of $B_1$ are fixed (Condition~\ref{cond:stvar_B_1} of Proposition~\ref{prop:stvar_ident}), we then obtain from the above results that almost everywhere in $[B_1:...:B_M] \in \mathbb{R}^{d\times dM}$, reordering or changing the signs of the columns of $B_2,...,B_M$ would, at some $t$, lead to an impact matrix that cannot be obtained by reordering or changing the signs of the columns of $B_{y,t}$. Hence, models before and after reordering or changing the signs of the columns of any of $B_2,...,B_M$ are not observationally equivalent. The ordering and signs of the columns of $B_1,...,B_M$ are thereby effectively fixed under the conditions of Proposition~\ref{prop:stvar_ident}, implying that the ordering and signs of the columns of $B_{y,t}$ are also fixed for all $t$. Thus, $B_{y,t}$ is uniquely identified at all $t$ almost everywhere in $[B_1:...:B_M] \in \mathbb{R}^{d\times dM}$. 

Consider the impact matrix $B_{y,t}$ based on invertible $(d\times d)$  matrices $B_1,...,B_M$, and suppose an impact matrix identical to $B_{y,t}$ for all $t$ is obtained with some other invertible $(d\times d)$ matrices $\tilde{B}_1,...,\tilde{B}_M$. Then,
\begin{equation}\label{eq:stvar_indeq4}
\sum_{m=1}^M \alpha_{m,t}\Delta B_m = 0,
\end{equation}
for all $t$, where $\Delta B_m \equiv B_m - \tilde{B}_m$.

In the case of logistic weights, $M=2$ so $\alpha_{1,t} = 1 - \alpha_{2,t}$ and Equation~(\ref{eq:stvar_indeq4}) reads
\begin{equation}\label{eq:stvar_indeq5}
(1 - \alpha_{2,t})\Delta B_1 + \alpha_{2,t}\Delta B_2 = 0.
\end{equation}
Because $0<\gamma<\infty$ and (endogenous) $z_t$ is real-valued and non-constant, the logistic map $\alpha_{2,t}$ attains at least two distinct values in $(0,1)$. Evaluating~(\ref{eq:stvar_indeq5}) at two such values, $\alpha^a\neq\alpha^b$, and subtracting gives $(\alpha^{a}-\alpha^{b})(\Delta B_{2}-\Delta B_{1})=0$, implying that $\Delta B_{2}=\Delta B_{1}$. Substituting back to~(\ref{eq:stvar_indeq5}) then yields $\Delta B_{1}=\Delta B_{2}=0$, i.e., $\tilde{B}_1=B_1$ and $\tilde{B}_2=B_2$. 

In the exogenous case (with any $M\geq 2$), by Condition~\ref{cond:linind} of Lemma~\ref{lemma:Bt_at_each_t} there exist $M$ indices $t_1,...,t_M$ such that the $(M\times M)$ weight matrix $W$
\begin{equation}
W\equiv \begin{bmatrix}
\alpha_{1,t_1} & \cdots & \alpha_{M,t_1} \\
\vdots         & \ddots & \vdots \\
\alpha_{1,t_M} & \cdots & \alpha_{M,t_M}
\end{bmatrix}
\end{equation}
is invertible. Writing Equation~(\ref{eq:stvar_indeq4}) at these $M$ indices gives the linear system 
\begin{equation}\label{eq:stvar_indeq6}
\sum_{m=1}^M\alpha_{m,t_i}\Delta B_m = 0, \quad i=1,...,M,
\end{equation}
which can be expressed as $(W\otimes I_{d^2})(\text{vec}(\Delta B_1),...,\text{vec}(\Delta B_M))=0$. Since $W$ is invertible, $W\otimes I_{d^2}$ has full rank, so $\Delta B_m = 0$, i.e., $\tilde{B}_m=B_m$, for all $m=1,...,M$. This concludes that the matrices $B_1,...,B_M$ are uniquely identified almost everywhere in $[B_1:...:B_M]\in\mathbb{R}^{d\times dM}$.\qed


\subsection{Proof of Proposition~\ref{prop:tvar_ident}}\label{sec:proofpropotvar}

We start by showing that under the conditions of Proposition~\ref{prop:tvar_ident}, the AR and threshold parameters are identified. Then, we combine this result with Lemma~\ref{lemma:Bt_at_each_t} to show that $B_1,...,B_M$ are identified up to ordering and signs of their columns. 

In the case of exogenous binary weights, unique identification of $\phi_{m},A_{m,1},...,A_{m,p}$, $m=1,...,M$, is given by Lemma~\ref{lemma:Bt_at_each_t} (whose conditions for exogenous weights are satisfied under conditions of Proposition~\ref{prop:tvar_ident}). Consider the threshold transition weights defined in~(\ref{eq:alpha_mt_threshold}) with the switching variable $z_t=a'\boldsymbol{y}_{t-1}$, where $\boldsymbol{y}_{t-1}=(y_{t-1},...,y_{t-p})$ and $a$ is a $(dp\times 1)$ vector with exactly one of its elements equal to one and others equal to zero. At each $t$, the model reduces to a linear VAR $y_t = \phi_{m} + \sum_{i=1}^pA_{m,i}y_{t-1} + B_me_t$ corresponding to the active regime $m$ for which $r_{m-1}<z_t\leq r_m$. Consider Equation~(\ref{eq:tildeproof2}) with two sets of parameters $\phi_{m},A_{m,1},...,A_{m,p}$, $m=1,...,M$, $r_1,...,r_{M-1}$ and $\tilde{\phi}_{m},\tilde{A}_{m,1},...,\tilde{A}_{m,p}$, $m=1,...,M$, $\tilde{r}_1,...,\tilde{r}_{M-1}$. Note that $r_0=\tilde{r}_0=-\infty, r_M=\tilde{r}_M=\infty$ and $-\infty<r_1< ... <r_{M-1}<\infty$ by assumption.

Denote $J_m\equiv\{z_t\in\mathbb{R}: r_{m-1}<z_t\leq r_m\}, \tilde{J}_n\equiv\{z_t\in\mathbb{R}:\tilde{r}_{n-1}<z_t\leq \tilde{r}_n\}$, and $\dot{J}_{m,n}\equiv J_m\cap \tilde{J}_n$ for $m,n=1,...,M$. Because $z_t$ has strictly positive density almost everywhere (by Assumption~\ref{as:shocks} and invertibility of the impact matrix in the active regime), each $J_m$ has positive Lebesgue measure. Consequently, since $\{\tilde{J}_n \}_{n=1}^M$ partition $\mathbb{R}$, for every $m=1,...,M$, there exists at least one $n$ where $\dot{J}_{m,n}$ has positive Lebesgue measure. In this overlap event $\alpha_{m,t}=\tilde{\alpha}_{n,t}=1$, so Equation~(\ref{eq:tildeproof2}) gives
\begin{equation}
(\tilde{\phi}_{n} - \phi_{m}) + (\tilde{\boldsymbol{A}}_n - \boldsymbol{A}_m)\boldsymbol{y}_{t-1}=0 \ \ \text{a.s. in} \ \dot{J}_{m,n},
\end{equation}
where $\boldsymbol{A}_m = [A_{m,1}:...:A_{m,p}]$ $(d\times dp)$ and $\tilde{\boldsymbol{A}}_n = [\tilde{A}_{n,1}:...:\tilde{A}_{n,p}]$ $(d\times dp)$. 

By Assumption~\ref{as:shocks} and invertibility of the impact matrix in the active regime, the lag vector $\boldsymbol{y}_{t-1}$ has strictly positive density almost everywhere in $\mathbb{R}^{dp}$, implying that it has strictly positive density also a.e. in $\{ \boldsymbol{y}_{t-1}\in \mathbb{R}^{dp}: z_t\in \dot{J}_{m,n} \}$. Since an affine function that equals zero a.e. on a set with positive Lebesgue measure must have all coefficients zero, we have $\tilde{\phi}_{n} = \phi_{m}$ and $\tilde{\boldsymbol{A}}_n = \boldsymbol{A}_m$ for $\dot{J}_{m,n}$ with positive measure. 

Now, if there are two positive measure sets $\dot{J}_{m,n_1}$ and $\dot{J}_{m,n_2}$ with $n_1\neq n_2$, it follows that $(\tilde{\phi}_{n_1},\tilde{\boldsymbol{A}}_{n_1})=(\tilde{\phi}_{n_2},\tilde{\boldsymbol{A}}_{n_2})$ violating Condition~\ref{cond:tvar_different_arpars} of Proposition~\ref{prop:tvar_ident}. Hence, for each $m$ there exists a unique $n=n(m)$ with $\dot{J}_{m,n(m)}$ that has positive measure. 
Since $\bigcup_{n=1}^M\dot{J}_{m,n}=J_m$ and only single $\dot{J}_{m,n(m)}$ can have positive Lebesgue measure, it follows that $J_m\subseteq \tilde{J}_{n(m)}$ almost everywhere. By symmetric arguments, we get that also $\tilde{J}_{n(m)} \subseteq J_m$ almost everywhere, so $J_m = \tilde{J}_{n(m)}$ almost everywhere. Because the symmetric difference between the half-open intervals $J_m$ and $\tilde{J}_{n(m)}$ 
has thereby Lebesgue measure zero, the end points must coincide, as otherwise the symmetric difference would contain a non-degenerate interval. 
Thus, $r_{m-1}=\tilde{r}_{n(m)-1}$ and $r_{m}=\tilde{r}_{n(m)}$. Since the thresholds are ordered by assumption so that regimes cannot be 'relabelled', it follows that $n(m)=m$. Hence, $(\tilde{\phi}_{m},\tilde{A}_{1,m},...,\tilde{A}_{p,m}) = (\phi_{m},A_{1,m},...,A_{p,m})$ for all $m=1,...,M$ and $\tilde{r}_{m}=r_{m}$ for all $m=1,...,M-1$, concluding that the AR and threshold parameters are uniquely identified. 

By Lemma~\ref{lemma:Bt_at_each_t}, $B_{y,t}$ is identified up to ordering and signs of its columns at each $t$, given that it is invertible. Since the transition weights are binary, $B_{y,t}=B_m$ for some $m$ for all $t$. Since each $B_m$ is assumed invertible, it follows that $B_{y,t}$ is invertible for all $t$. Hence, for each regime $m$, the associated impact matrix $B_m$ is identified up to ordering and signs of its columns, concluding the proof.\qed

\subsection{Proof of Theorem~\ref{thm:stat}}\label{sec:proofthmstat}

As explained in Appendix~\ref{sec:stat}, our STVAR model is not nested to the STVEC model of \cite{Saikkonen:2008}, and hence, additional steps are needed to show the applicability of the results of \cite{Saikkonen:2008} to our structural STVAR model. Note that by Assumption~\ref{as:stat}\ref{cond:statshock}, Assumption~1 in \cite{Saikkonen:2008} is satisfied. Moreover, by the constraint $\sum_{m=1}^M\alpha_{m,t}=1$, Assumption~2(a) in \cite{Saikkonen:2008} satisfied, and as we assume that the set $\text{l}_2$ in \cite{Saikkonen:2008} is an empty set, also Assumption~2(b) is satisfied. Consequently, Condition~(10) in \cite{Saikkonen:2008} is satisfied and by our Assumption~\ref{as:stat}\ref{cond:jsr} also Condition~(19). The difference of our setup to the conditions of Theorem~1 of \cite{Saikkonen:2008} is that our STVAR model specifies the conditional covariance matrix differently to the STVEC model defined in Equation~(8) of \cite{Saikkonen:2008}. 

Besides the conditional covariance matrix, our structural STVAR model defined in Equations~(\ref{eq:stvarstruct}) and (\ref{eq:Bt}) is obtained from the stochastic process $z_t$ defined in Equation~(17) (or equally (15)) of \cite{Saikkonen:2008} as follows ($z_t$ of \cite{Saikkonen:2008} should not be confused with the switching variable of this paper). First, we assume that the number of unit roots in the model is zero, so $n=r$ in \cite{Saikkonen:2008}. Consequently, $z_t=y_t$ and, in the definition of the matrix $J$ defined above Equation~(17) of \cite{Saikkonen:2008}, $\beta=c=I_d$ \citep[see][Note~2]{Saikkonen:2008}. The matrix $J$ is thereby nonsingular, and following \cite{Kheifets+Saikkonen:2020}, we can choose the matrix $S$ in Equation~(17) to be the inverse of $J'$. Second, we assume the transition functions $h_s(S'JZ_{t-1},\eta_t)$ of \cite{Saikkonen:2008} are independent of the random variable $\eta_t$. Hence, the transition functions are functions of $Z_{t-1}=(y_{t-1},...,y_{t-p})$, and we can write $h_s(S'JZ_{t-1},\eta_t)=\alpha_{m,t}$. Using the notation of this paper, the first term on the right side of Equation~(17) in \cite{Saikkonen:2008} then reduces to $\sum_{m=1}^M\alpha_{m,t}\sum_{i=1}^pA_{m,i}y_{t-i}$. Third, assume the set $\text{l}_2$ is empty, which implies that the second term on the right side of Equation~(17), defined in Equation~(11) of \cite{Saikkonen:2008}, reduces to $\sum_{m=1}^M\alpha_{m,t}\phi_{m}$ (with the notation of this paper).

That is, the first two terms on the right side of Equation~(17) in \cite{Saikkonen:2008}, simplify to the first term on the right hand side of Equation~(\ref{eq:stvarstruct}) defining our structural STVAR model. The third term on the right side of Equation~(17) in \cite{Saikkonen:2008}, however, 
simplifies to $(\sum_{m=1}^M\alpha_{m,t}\Omega_m)^{1/2}e_t$, where $\Omega_1,...,\Omega_M$ are positive definite $(d\times d)$ covariance matrices of the regimes. This is different to the second term of Equation~(\ref{eq:stvarstruct}), which takes the form $\sum_{m=1}^M\alpha_{m,t}B_me_t$ when the impact matrix~(\ref{eq:Bt}) is substituted to it. Nonetheless, it turns out that with only slight modifications, the proof of Theorem~1 in \cite{Saikkonen:2008} applies to our model as well. 

The conditional covariance matrix is relevant in only a few places in the proof of Theorem~1 in \cite{Saikkonen:2008}. Firstly, on Page~312, \cite{Saikkonen:2008} makes use of the positive definiteness of the conditional covariance matrix $H(S'JZ_{t-1})$ (or $\sum_{m=1}^M\alpha_{m,t}\Omega_m$ in our notation) to conclude that the smallest eigenvalue of $H(S'Jx)$, $x\in\mathbb{R}^{dp}$, is bounded away from zero in compact subsets of $\mathbb{R}^{dp}$. \cite{Saikkonen:2008} argues that thereby by Theorem~2.2(ii) of \cite{Cline+Phu:1998}, the Markov chain $Z_t$ is an irreducible and aperiodic T-chain on $\mathbb{R}^{dp}$, which is then used to conclude further results. 

Concerning our model, note that by Assumption~\ref{as:stat}\ref{cond:statweights}, the regime weights are either of the threshold or logistic form. In the threshold case, denoting a generic lag vector as $\boldsymbol{y}\in\mathbb{R}^{dp}$, the impact matrix $B_{y,t}(\boldsymbol{y})=\sum_{m=1}^M\alpha_{m,t}(\boldsymbol{y})B_m$ is exactly $B_m$ within each regime, and since each $B_m$ is assumed invertible, the conditional covariance matrix $B_{y,t}(\boldsymbol{y})B_{y,t}(\boldsymbol{y})’$ is positive definite. In the logistic two-regime case, Condition~\ref{cond:logB1B2} rules out singular convex combinations of $B_1$ and $B_2$, which likewise guarantees that $B_{y,t}(\boldsymbol{y})$ is invertible for all $\boldsymbol{y}$.\footnote{Indeed, 
 denoting $C\equiv B_1^{-1}B_2$ and $B(\alpha)=(1-\alpha)B_1+\alpha B_2 = B_1((1-\alpha)I+\alpha C)$, the following conditions are equivalent: (1) $B(\alpha)$ is singular for some $\alpha^*\in(0,1)$ and (2) $C$ has a negative real eigenvalue. To prove this, observe first that $\det B(\alpha)=\det(B_1)\det((1-\alpha)I+\alpha C)$. As $B_1$ is invertible, it follows that $B(\alpha)$ is singular iff $(1-\alpha)I+\alpha C$ is singular, i.e., iff there exists $v\neq 0$ with $((1-\alpha)I+\alpha C)v=0$, which is equivalent to $Cv = -\frac{1-\alpha}{\alpha}v$. For $\alpha\in(0,1)$, the scalar $-\frac{1-\alpha}{\alpha}$ is real and negative. Hence, if $B(\alpha^*)$ is singular for some $\alpha^*\in(0,1)$, then $C$ has a negative real eigenvalue $-\tfrac{1-\alpha^*}{\alpha^*}$. Conversely, if $Cv=\lambda v$ with $\lambda<0$, choose $\alpha^*=\frac{1}{1-\lambda}\in(0,1)$. Then, $(1-\alpha^*)+\alpha^*\lambda=0$, so $(1-\alpha^*)I+\alpha^* C$ (and hence $B(\alpha^*))$ is singular, establishing the equivalence.}
Therefore, and because in the logistic case $B_{y,t}(\boldsymbol{y})$ is continuous in $\boldsymbol{y}$, the smallest eigenvalue of $B_{y,t}(\boldsymbol{y})B_{y,t}(\boldsymbol{y})’$ attains a strictly positive minimum on every compact subset of $\mathbb{R}^{dp}$. Thus, the conditional covariance matrix is uniformly bounded away from singularity on compact subsets as required. It follows from this, Equation~(\ref{eq:stvarstruct}), and the invertibility of $B_{y,t}$ that we can apply Theorem~2.2(ii) of \cite{Cline+Phu:1998} to conclude that the Markov chain $\boldsymbol{y}_{t}$ is an irreducible and aperiodic T-chain on $\mathbb{R}^{dp}$, similarly to the process $Z_{t-1}$ in \cite{Saikkonen:2008}. 

Secondly, the conditional covariance matrix appears in the proof of \cite{Saikkonen:2008} on Page~313, 
where \cite{Saikkonen:2008} shows that an appropriate version of Condition (15.3) of \citet[][p. 358]{Meyn+Tweedie:1993} is satisfied. To show that this condition holds for our model, we need to show that an appropriate version of Equation~(A.5) in \cite{Saikkonen:2008} holds for our model, i.e., that the term $E[\sup_{A\in \mathcal{A}^j}||AB_{y,t}(\boldsymbol{y})e_t||^2]$ is bounded from above by a finite real number, where $j$ is a positive integer, $\mathcal{A}=\lbrace \boldsymbol{A}_1,...,\boldsymbol{A}_M\rbrace$, and $\boldsymbol{A}_1,...,\boldsymbol{A}_M$ as well as the set $\mathcal{A}^j$ are defined in Appendix~\ref{sec:stat}. 

Similarly to \citet[][p. 313]{Saikkonen:2008}, clearly,
\begin{equation}\label{eq:supbound}
\sup\lbrace ||A||:A\in \cup_{j=1}^{N}\mathcal{A}^j  \rbrace < c,
\end{equation}
where $||\cdot||$ is Euclidean norm, $N$ is a positive integer, and $c$ a finite real number. By Jensen's inequality and submultiplicativity of norms, we have
\begin{equation}
E\left[\sup_{A\in \mathcal{A}^j}||AB_{y,t}(\boldsymbol{y})e_t||^2\middle| \boldsymbol{y}=x \right] \leq \sup_{A\in \mathcal{A}^j}||A||^2||B_{y,t}(x)||^2E||e_t||^2, 
\end{equation} 
where the right side is bounded from above by some finite constant (by~(\ref{eq:supbound}) and the finite second moment of $e_t$ given by Assumption~\ref{as:shocks}). Thus, an appropriate version of Equation~(A.5) in \cite{Saikkonen:2008} holds for our model. Since the rest of the proof of Theorem~1 in \cite{Saikkonen:2008} thereby applies to our model as well, this concludes the proof.\qed

\section{Penalized nonlinear least squares estimation of STVAR models}\label{sec:nls}
This appendix describes how the parameters of the structural STVAR model discussed in Section~\ref{sec:structstvar} can be estimated by the method of nonlinear least squares (NLS) as explained in \citet[][Section~4.2]{Hubrich+Terasvirta:2013} for logistic STVAR models. After describing the NLS estimation, penalized nonlinear least squares (PNLS) estimation is considered by introducing a penalization term to the sum of squares of residuals function. 

\subsection{Nonlinear least squares estimation}\label{sec:subnls}
Consider the following sum of squares of residuals function
\begin{equation}\label{eq:nlm_Q}
Q(\alpha) = \sum_{t=1}^T u_t(\alpha)'u_t(\alpha),
\end{equation}
where $u_t(\alpha) = y_t - \sum_{m=1}^M\alpha_{m,t}\mu_{m,t}$ and the parameter $\alpha$ contains the weight function parameters.\footnote{With threshold weights~(\ref{eq:alpha_mt_threshold}), $\alpha=(r_1,...,r_{M-1})$, and with logistic weights~(\ref{eq:alpha_mt_logistic}) weights, $\alpha=(c,\gamma)$. If the transition weights are exogenous, we simply treat them as known constants.} The NLS estimates of the parameters $\phi_{m},A_{m,1},...,A_{m,p}$, $m=1,...,M$, corresponding to the weight function parameter $\alpha$ are obtained by minimizing $Q(\alpha)$ with respect to these parameters. 

If the transition weights $\alpha_{1,t},...,\alpha_{M,t}$ are not exogenous, also the parameter $\alpha$ they depend on needs to be estimated. Since the NLS estimates of the AR parameters are analytically obtained for given $\alpha$, the optimization can be performed with respect to $\alpha$ only. However, this optimization problem is multimodal, and thereby standard gradient based optimization methods that are not able escape from local maxima are insufficient. Since we only need approximate estimates to serve as starting values for Step~\ref{step:GA} and~\ref{step:VA} estimation in the three-step procedure described in Section~\ref{sec:three-step}, the following procedure is adopted. 

First, determine the values of the transition weights $\alpha_{m,t}$, $m=1,...,M$, $t=1,...,T$, that can be deemed reasonable for estimation, i.e., transition weights that allocate a sufficient enough contribution to each regime in the model. As an extreme example, transition weights with $\alpha_{1,t}=0$ for all $t=1,...,T$ are not reasonable, as Regime~1 would not be involved in the model at all. Thus, each regime should contribute to the model to an extend that is reasonable enough for estimation. 

One way to assess the contribution of each regime is by summing over the transition weights, i.e., calculating the sums $\dot{\alpha}_m \equiv \sum_{t=1}^T\alpha_{m,t}$, $m=1,...,M$. In the special case of discrete transition weights such as the threshold weights~(\ref{eq:alpha_mt_threshold}), the sums $\dot{\alpha}_m$ express the numbers of observations from each regime. It can then be determined whether each regime contributes enough to the model by checking, for all $m=1,...,M$, whether the sums $\dot{\alpha}_m$ are large enough compared to the number of parameters in that regime. Denoting the large enough sum as $\tilde{T}$, transition weights that satisfy $\dot{\alpha}_m\geq \tilde{T}$ for all $m=1,...,M$ are deemed reasonable enough.\footnote{In our empirical application (Section~\ref{sec:empapp}), we assume that there should be at least three times as many observations from all the $d$ variables combined, implying that $\tilde{T}=3kd^{-1}$, where $k$ is the number of parameters in a regime. We acknowledge that this assumption 
is somewhat arbitrary, but it attempts to deter obvious cases of overfitting, while still allowing flexibility when the model is large compared to the number of observations.}

Second, specify a range of reasonable parameter values for each scalar parameter in the $(s \times 1)$ vector $\alpha$, and specify a grid of parameter values within that range. For instance, for a logistic STVAR model, the location parameter should, at least, lie between the lowest and greatest value in the observed series of the switching variable, whereas the scale parameter can range from a small positive number to some large positive number that implies almost discrete regime-switches. With $n$ grid points for each $s$ scalar parameters in $\alpha$, the number of different parameter vectors $\alpha$ with values in the grid points is $n^s$. Hence, we recommend selecting a sufficiently small number of grid points $n$, so that estimation can be performed in a reasonable computation time. Denote the $i$th vector of $\alpha$ containing parameter values in the grid points as $\alpha^{(i)}$, $i=1,...,n^s$. Then, obtain the subset of the grid points $\alpha^{(i)}$ that satisfy $\dot{\alpha}_m\geq \tilde{T}$ for all $m=1,...,M$. Denote this subset as $\mathcal{W}$. 

Finally, calculate the NLS estimates of the AR parameters for all $\alpha^{(i)}\in\mathcal{W}$, and the estimates yielding the smallest sum of squares of residuals $Q(\alpha)$ are the approximate NLS estimates of the parameters $\phi_{m},A_{m,1},...,A_{m,p}$, $m=1,...,M$, and $\alpha$. Notably, the NLS estimation does not necessarily produce estimates that satisfy the usual stability condition for all or any of the regimes.

\subsection{Penalized nonlinear least squares estimation}

Penalized nonlinear least squares estimation of STVAR models works similarly to the NLS estimation, except that instead of the sum of squares function~(\ref{eq:nlm_Q}), the penalized sum of squares function is minimized. We define the penalized sum of squares function as
\begin{equation}\label{eq:penQ}
PQ(\alpha) = Q(\alpha) - P(\alpha),
\end{equation}
where the penalization term $P(\alpha)$ increases as the parameter values approach the boundary of, or move further into, an uninteresting region of the parameter space. To focus on avoiding estimates that do not satisfy the usual stability condition for all regimes, we define the penalty term as 
\begin{equation}\label{eq:pentermLS}
P(\alpha) = \kappa \widehat{RSS} \sum_{m=1}^M\sum_{i=1}^{dp} \text{max}\lbrace 0, |\rho(\boldsymbol{A}_m(\boldsymbol{\theta}))_i| - (1 - \eta) \rbrace^2,
\end{equation}
where 
\begin{equation}
\widehat{RSS}=\underset{\alpha\in\mathcal{W}}{\min}Q(\alpha)
\end{equation}
is the minimized residual sum of squares, $|\rho(\boldsymbol{A}_m(\boldsymbol{\theta}))_i|$ is the modulus of the $i$th eigenvalue of the companion form AR matrix of regime~$m$, the tuning parameter $\eta\in (0, 1)$ determines how close to the boundary of the stability region the penalization starts, and $\kappa>0$ determines the strength of the penalization. Whenever the companion form AR matrix of a regime has eigenvalues greater than $1 - \eta$ in modulus, the penalization term~(\ref{eq:pentermLS}) is greater than zero, and it increases in the modulus of these eigenvalues. The penalty term~(\ref{eq:pentermLS}) is multiplied by $\widehat{RSS}$ to standardize the strength of the penalization with respect to the fit of the model.  

The penalized NLS estimates of $\phi_{m},A_{m,1},...,A_{m,p}$, $m=1,...,M$, and $\alpha$, are obtained by first minimizing $Q(\alpha)$ for all $\alpha\in\mathcal{W}$ (where $\mathcal{W}$ is the set of grid point defined in Section~\ref{sec:subnls}), then calculating the penalty term $P(\alpha)$ for all $\alpha\in\mathcal{W}$, and finally selecting the estimates that give the smallest minimized the penalized sum of squares~(\ref{eq:penQ}). In our empirical analysis in Section~\ref{sec:empapp}, we use the tuning parameter values $\eta=0.05$ and $\kappa=0.2$, in line with the tuning parameter values used in the PML estimation (see Section~\ref{sec:loglik}).

\section{A Monte Carlo study}\label{sec:montecarlo}
This appendix present a Monte Carlo study that assesses the performance of the proposed penalized maximum likelihood estimator and the three-step estimation procedure discussed in Section~\ref{sec:estimation}. 
%
%
We consider two specifications of a simple bivariate (structural) LSTVAR model with $p=1$ and $M=2$, with the first lag of the first variable as the switching variable. To maintain some constancy, we use the same impact matrices across both specifications, $\text{vec}(B_1)=(0.6, -0.3, 0.2, 0.4)$ and $\text{vec}(B_2)=(0.7, 0.1, 0.3, 0.8)$, the same location and scale parameter values $c=0.8$ and $\gamma=5.0$, as well as the same degrees-of-freedom and skewness parameter values, $\nu_1=2.5$, $\nu_2=12.0$, $\lambda_1=-0.5$, and $\lambda_2=0.2$.

We consider two specifications of the AR parameters, presented in Table~\ref{tab:monteparams}. In the first specification (LSTVAR~1), the AR matrices are inside the stability region (see Section~\ref{sec:ml_stat}). In the second specification (LSTVAR~2), the AR matrices are in the region where penalization term is strictly positive, allowing us to assess the performance of our estimation method when the true parameter value lies in the penalization region. To illustrate the closeness to the boundary, the modulus of the roots of the companion form AR matrix are presented in Table~\ref{tab:monteparams}. As in our empirical application, we use the tuning parameter values $\eta = 0.05$ and $\kappa=0.2$ for the penalization term, so the penalization term is strictly positive when at least one of the roots of the companion form AR matrix is larger than $0.95$ in modulus. The intercept parameter values are adjusted so that the unconditional means of the regimes are approximately equal in both of the specifications (see Table~\ref{tab:monteparams}), so that the main difference between the specifications is that the second one includes true parameter values from inside the penalization region. 

\begin{table}[t]
\footnotesize
\centering
\begin{tabular}{c c c@{\hspace{3pt}}c c c c c@{\hspace{3pt}}c c c}
 & \multicolumn{5}{c}{Regime 1} & \multicolumn{5}{c}{Regime 2} \\
 & $\phi_{1}$ & \multicolumn{2}{c}{$A_{1,1}$} & $\mu_1$ & $||\boldsymbol{A}_1||$ & $\phi_{2}$ &  \multicolumn{2}{c}{$A_{2,1}$} & $\mu_2$ & $||\boldsymbol{A}_2||$ \\ 
\hline\\[-1.0ex]
\multirow{2}{*}{LSTVAR~1} & $0.30$ & $0.70$ & $-0.30$           & $0.00$ & $0.58$ & $\phantom{-}1.20$ & $0.50$ & $0.20$ & $\phantom{-}2.00$ & $0.74$ \\
                                    & $0.60$ & $0.20$ & $\phantom{-}0.40$ & $1.00$ & $0.58$ & $-1.10$           & $0.30$ & $0.50$ & $-1.00$ & $0.26$ \\
\hline\\[-1.0ex]
\multirow{2}{*}{LSTVAR~2} & $0.30$ & $1.10$ & $-0.30$           & $0.00$ & $0.97$ & $\phantom{-}0.72$ & $0.74$ & $0.20$ & $\phantom{-}2.00$ & $0.98$ \\
                                    & $0.20$ & $0.20$ & $\phantom{-}0.80$ & $1.00$ & $0.97$ & $-0.87$           & $0.30$ & $0.73$ & $-1.00$ & $0.49$ \\
\end{tabular}
\caption{The parameters values of $\phi_{1},\phi_{2},A_{1,1}$, and $A_{2,1}$ used for the LSTVAR models LSTVAR~1 and LSTVAR~2 in the Monte Carlo study. Also, the corresponding unconditional means of the regimes $\mu_m=(I_d - A_{m,1})^{-1}\phi_{m}$, $m=1,2$, and modulus of the eigenvalues of the companion form AR matrices of the regimes $||\boldsymbol{A}_m||$, $m=1,2$.}
\label{tab:monteparams}
\end{table}

For each of the two model specifications, we generate $500$ samples of length $250,500,1000,2000$, and $10000$. Then, we fit the first-order two-regime LSTVAR model to each of the generated sample using the three-step estimation procedure based on the penalized log-likelihood function, proposed in Section~\ref{sec:estimation}. 
After estimation, the ordering and signs of the columns of $B_1$ and $B_2$ are standardized by assuming, without loss of generality, that the degrees-of-freedom parameters are in an ascending order and that the first skewness parameter is negative and the second one positive, in line with the true parameter values. %
%
To estimate the bias of the PML estimator, we calculate the average estimation errors $\hat{\boldsymbol{\theta}} - \boldsymbol{\theta}$ over the $500$ Monte Carlo repetitions, and to estimate the spread of the estimates, we calculate their sample standard deviations.

\begin{table}[p]
\footnotesize
\renewcommand{\arraystretch}{0.85}
\centering
\begin{tabular}{c c @{\hspace{8pt}} d{0}@{\hspace{0pt}}d{0} @{\hspace{14pt}} d{0}@{\hspace{0pt}}d{0} @{\hspace{14pt}} d{0}@{\hspace{0pt}}d{0} @{\hspace{14pt}} d{0}@{\hspace{8pt}}d{0} @{\hspace{14pt}} d{0}@{\hspace{8pt}}d{0}}
& & \multicolumn{2}{c}{$T=250$}\hspace{4pt} & \multicolumn{2}{c}{$T=500$}\hspace{8pt} & \multicolumn{2}{c}{$T=1000$}\hspace{10pt} & \multicolumn{2}{c}{$T=2000$}\hspace{12pt} & \multicolumn{2}{c}{$T=10000$}\\ 
\hline\\[-1.7ex]
LSTVAR~1 & $\phi_{1}$   & $0.01$ & $(0.17)$ & $0.01$ & $(0.08)$ & $0.01$ & $(0.06)$ & $0.00$ & $(0.04)$ & $0.00$ & $(0.02)$ \\
         &              & $0.00$ & $(0.25)$ & $0.00$ & $(0.12)$ & $0.00$ & $(0.07)$ & $0.00$ & $(0.05)$ & $0.00$ & $(0.02)$ \\
         & $\phi_{2}$   & $-0.01$ & $(0.17)$ & $-0.01$ & $(0.11)$ & $0.00$ & $(0.08)$ & $0.00$ & $(0.06)$ & $0.00$ & $(0.03)$ \\
         &              & $-0.02$ & $(0.39)$ & $-0.02$ & $(0.24)$ & $-0.01$ & $(0.15)$ & $0.00$ & $(0.11)$ & $0.00$ & $(0.05)$ \\
         & $A_{1,1}$    & $0.00$ & $(0.15)$ & $0.01$ & $(0.07)$ & $0.00$ & $(0.04)$ & $0.00$ & $(0.03)$ & $0.00$ & $(0.01)$ \\
         &              & $-0.02$ & $(0.26)$ & $-0.01$ & $(0.13)$ & $0.00$ & $(0.06)$ & $0.00$ & $(0.04)$ & $0.00$ & $(0.01)$ \\
         &              & $0.02$ & $(0.08)$ & $0.01$ & $(0.05)$ & $0.00$ & $(0.03)$ & $0.00$ & $(0.02)$ & $0.00$ & $(0.01)$ \\
         &              & $-0.01$ & $(0.10)$ & $-0.01$ & $(0.06)$ & $0.00$ & $(0.04)$ & $0.00$ & $(0.03)$ & $0.00$ & $(0.01)$ \\
         & $A_{2,1}$    & $0.00$ & $(0.06)$ & $0.00$ & $(0.04)$ & $0.00$ & $(0.03)$ & $0.00$ & $(0.02)$ & $0.00$ & $(0.01)$ \\
         &              & $0.00$ & $(0.14)$ & $0.00$ & $(0.09)$ & $0.00$ & $(0.06)$ & $0.00$ & $(0.04)$ & $0.00$ & $(0.02)$ \\
         &              & $-0.01$ & $(0.04)$ & $0.00$ & $(0.03)$ & $0.00$ & $(0.02)$ & $0.00$ & $(0.01)$ & $0.00$ & $(0.01)$ \\
         &              & $-0.01$ & $(0.08)$ & $-0.01$ & $(0.05)$ & $0.00$ & $(0.04)$ & $0.00$ & $(0.02)$ & $0.00$ & $(0.01)$ \\
         & $B_1$        & $0.24$ & $(0.60)$ & $0.10$ & $(0.37)$ & $0.04$ & $(0.20)$ & $0.02$ & $(0.12)$ & $0.01$ & $(0.04)$ \\
         &              & $-0.09$ & $(0.46)$ & $-0.05$ & $(0.24)$ & $-0.02$ & $(0.13)$ & $-0.01$ & $(0.08)$ & $0.00$ & $(0.03)$ \\
         &              & $-0.03$ & $(0.12)$ & $-0.02$ & $(0.08)$ & $-0.01$ & $(0.04)$ & $0.00$ & $(0.03)$ & $0.00$ & $(0.01)$ \\
         &              & $-0.05$ & $(0.14)$ & $-0.03$ & $(0.08)$ & $-0.01$ & $(0.05)$ & $0.00$ & $(0.03)$ & $0.00$ & $(0.01)$ \\
         & $B_2$        & $0.44$ & $(0.72)$ & $0.19$ & $(0.42)$ & $0.08$ & $(0.22)$ & $0.04$ & $(0.12)$ & $0.01$ & $(0.04)$ \\
         &              & $0.05$ & $(0.17)$ & $0.02$ & $(0.09)$ & $0.01$ & $(0.05)$ & $0.01$ & $(0.03)$ & $0.00$ & $(0.01)$ \\
         &              & $-0.01$ & $(0.07)$ & $0.00$ & $(0.02)$ & $0.00$ & $(0.01)$ & $0.00$ & $(0.01)$ & $0.00$ & $(0.00)$ \\
         &              & $-0.03$ & $(0.19)$ & $0.00$ & $(0.03)$ & $0.00$ & $(0.02)$ & $0.00$ & $(0.02)$ & $0.00$ & $(0.01)$ \\
         & $c$          & $0.03$ & $(0.19)$ & $0.00$ & $(0.07)$ & $0.00$ & $(0.04)$ & $0.00$ & $(0.02)$ & $0.00$ & $(0.01)$ \\
         & $\gamma$     & $28.86$ & $(121.66)$ & $2.90$ & $(25.52)$ & $0.09$ & $(0.75)$ & $0.05$ & $(0.47)$ & $0.01$ & $(0.19)$ \\
         & $\nu_1$      & $-0.02$ & $(0.49)$ & $-0.02$ & $(0.32)$ & $-0.02$ & $(0.21)$ & $-0.01$ & $(0.15)$ & $0.00$ & $(0.07)$ \\
         & $\nu_2$      & $152.80$ & $(597.75)$ & $43.27$ & $(266.76)$ & $8.00$ & $(52.02)$ & $1.46$ & $(5.55)$ & $0.11$ & $(1.37)$ \\
         & $\lambda_1$  & $-0.04$ & $(0.09)$ & $-0.02$ & $(0.05)$ & $-0.01$ & $(0.04)$ & $0.00$ & $(0.02)$ & $0.00$ & $(0.01)$ \\
         & $\lambda_2$  & $0.02$ & $(0.10)$ & $0.01$ & $(0.07)$ & $0.00$ & $(0.05)$ & $0.00$ & $(0.03)$ & $0.00$ & $(0.01)$ \\
\hline\\[-1.7ex]
LSTVAR~2 & $\phi_{1}$   & $-0.01$ & $(0.18)$ & $-0.01$ & $(0.11)$ & $-0.01$ & $(0.08)$ & $-0.01$ & $(0.05)$ & $0.00$ & $(0.02)$ \\
         &              & $-0.10$ & $(0.19)$ & $-0.06$ & $(0.13)$ & $-0.03$ & $(0.08)$ & $-0.02$ & $(0.06)$ & $-0.02$ & $(0.02)$ \\
         & $\phi_{2}$   & $0.07$  & $(0.67)$ & $0.01$  & $(0.11)$ & $0.00$  & $(0.07)$ & $0.00$  & $(0.05)$ & $0.00$  & $(0.02)$ \\
         &              & $0.12$  & $(1.18)$ & $0.02$  & $(0.19)$ & $0.01$  & $(0.13)$ & $0.01$  & $(0.08)$ & $0.01$  & $(0.04)$ \\
         & $A_{1,1}$    & $-0.05$ & $(0.13)$ & $-0.02$ & $(0.07)$ & $-0.01$ & $(0.03)$ & $-0.01$ & $(0.02)$ & $0.00$  & $(0.01)$ \\
         &              & $-0.03$ & $(0.11)$ & $-0.02$ & $(0.08)$ & $-0.01$ & $(0.04)$ & $0.00$  & $(0.02)$ & $0.00$  & $(0.01)$ \\
         &              & $0.03$  & $(0.13)$ & $0.01$  & $(0.08)$ & $0.00$  & $(0.04)$ & $0.00$  & $(0.02)$ & $0.00$  & $(0.01)$ \\
         &              & $-0.03$ & $(0.08)$ & $-0.03$ & $(0.06)$ & $-0.02$ & $(0.03)$ & $-0.01$ & $(0.02)$ & $-0.01$ & $(0.01)$ \\
         & $A_{2,1}$    & $-0.02$ & $(0.11)$ & $-0.01$ & $(0.03)$ & $0.00$  & $(0.02)$ & $0.00$  & $(0.01)$ & $0.00$  & $(0.01)$ \\
         &              & $-0.02$ & $(0.21)$ & $-0.01$ & $(0.06)$ & $-0.01$ & $(0.04)$ & $0.00$  & $(0.03)$ & $0.00$  & $(0.01)$ \\
         &              & $-0.01$ & $(0.06)$ & $0.00$  & $(0.03)$ & $0.00$  & $(0.02)$ & $0.00$  & $(0.01)$ & $0.00$  & $(0.01)$ \\
         &              & $-0.03$ & $(0.13)$ & $0.00$  & $(0.05)$ & $0.00$  & $(0.03)$ & $0.00$  & $(0.02)$ & $0.00$  & $(0.01)$ \\
         & $B_1$        & $0.44$  & $(0.74)$ & $0.16$  & $(0.41)$ & $0.08$  & $(0.24)$ & $0.03$  & $(0.12)$ & $0.01$  & $(0.05)$ \\
         &              & $-0.14$ & $(0.55)$ & $-0.06$ & $(0.27)$ & $-0.04$ & $(0.18)$ & $-0.02$ & $(0.09)$ & $0.00$  & $(0.03)$ \\
         &              & $-0.04$ & $(0.13)$ & $-0.02$ & $(0.08)$ & $-0.01$ & $(0.06)$ & $-0.01$ & $(0.03)$ & $0.00$  & $(0.01)$ \\
         &              & $-0.03$ & $(0.19)$ & $-0.02$ & $(0.09)$ & $-0.02$ & $(0.05)$ & $-0.01$ & $(0.04)$ & $0.00$  & $(0.02)$ \\
         & $B_2$        & $0.64$  & $(0.84)$ & $0.25$  & $(0.44)$ & $0.12$  & $(0.25)$ & $0.05$  & $(0.12)$ & $0.01$  & $(0.04)$ \\
         &              & $0.08$  & $(0.19)$ & $0.03$  & $(0.09)$ & $0.01$  & $(0.05)$ & $0.01$  & $(0.03)$ & $0.00$  & $(0.01)$ \\
         &              & $-0.01$ & $(0.09)$ & $0.00$  & $(0.02)$ & $0.00$  & $(0.01)$ & $0.00$  & $(0.01)$ & $0.00$  & $(0.00)$ \\
         &              & $-0.04$ & $(0.23)$ & $-0.01$ & $(0.03)$ & $0.00$  & $(0.02)$ & $0.00$  & $(0.02)$ & $0.00$  & $(0.01)$ \\
         & $c$          & $0.23$  & $(0.80)$ & $0.05$  & $(0.31)$ & $0.00$  & $(0.06)$ & $0.00$  & $(0.04)$ & $0.00$  & $(0.01)$ \\
         & $\gamma$     & $79.82$ & $(349.54)$ & $22.77$ & $(90.88)$ & $5.71$  & $(44.05)$ & $0.39$  & $(4.28)$ & $0.10$  & $(0.30)$ \\
         & $\nu_1$      & $-0.11$ & $(0.57)$ & $-0.07$ & $(0.30)$ & $-0.05$ & $(0.21)$ & $-0.03$ & $(0.15)$ & $0.00$  & $(0.07)$ \\
         & $\nu_2$      & $192.21$& $(889.55)$ & $38.30$ & $(233.50)$& $9.58$  & $(79.39)$ & $1.37$  & $(4.71)$ & $0.15$  & $(1.42)$ \\
         & $\lambda_1$  & $-0.05$ & $(0.09)$ & $-0.02$ & $(0.06)$ & $-0.01$ & $(0.04)$ & $-0.01$ & $(0.03)$ & $0.00$  & $(0.01)$ \\
         & $\lambda_2$  & $0.00$  & $(0.11)$ & $0.00$  & $(0.07)$ & $0.00$  & $(0.05)$ & $0.00$  & $(0.03)$ & $0.00$  & $(0.01)$ \\
\end{tabular}
\caption{Results for the structural LSTVAR models from the Monte Carlo study assessing the performance of our penalized log-likelihood based three-step estimation method based on the PML estimator discussed in Section~\ref{sec:estimation}. The average estimation error is reported for the samples of length $250$, $500$, $1000$, $2000$, and $10000$ in each column first, and next to the biases are the standard deviations of the estimates in parentheses. The results are based on $500$ Monte Carlo repetitions.}
\label{tab:monteresults_lstvar}
\end{table}

The results of the Monte Carlo study are presented in Table~\ref{tab:monteresults_lstvar}. They show that there is some but mostly small bias, which diminishes when the sample size increases. Since also the standard deviation of the estimates diminishes with the sample size, our results are in line with (possible) consistency of the PML estimator. Nonetheless, for certain parameters, there appears to be a notable bias in small samples. 

First, the estimator of the second degrees-of-freedom parameter (true parameter value $\nu_2=12$) is substantially biased upwards in samples of length $T=250,500$, and $1000$, and also the standard deviation of the estimates is very large. This can explained by the property of the skewed $t$-distribution that when the degrees-of-freedom parameter value is already large, increasing it more has virtually no effect on the shape of the distribution. Consequently, if a large degrees-of-freedom parameter fits the data, the numerical optimization algorithms are incapable of appropriately discriminating between large and very large parameter values. Hence, possibly very large degrees-of-freedom parameter values drawn by random in Step~\ref{step:GA} estimation may remain large through Step~\ref{step:VA} estimation. However, the large upward bias does not seem concerning, as there is virtually no difference in the economic interpretations of large and very large degrees-of-freedom parameter values (with the latter amplifying the bias). 

Second, there is a substantial upward bias also for the scale parameter $\gamma$ when the sample size is small. A large scale parameter value indicates almost discrete regime-switches, suggesting that the PML estimator is, in small samples, biased towards estimating too fast regime-switches when they are relatively smooth in the true data generating process ($\gamma=5$). The upward bias is possibly heightened by the property that when $\gamma$ is large, the regime-switches are already nearly discrete, making it difficult for numerical estimation algorithms to properly distinguish between large and very large values, particularly in small samples. 

Third, the estimates of the first element of the first row of $B_1$ and $B_2$ exhibit a notable upward bias when the sample is small. A possible explanation is that the bias is related to the heavy tails of the skewed $t$-distribution, as the true degrees-of-freedom parameter value is small ($\nu_1=2.5$ for the related shock). 
In particular, the estimates may be sensitive to extreme values drawn for the heavy-tailed skewed $t$ shocks, potentially amplifying small-sample bias. 

In general, the small-sample bias of the PML estimator seems to be slightly larger the specification where the true AR matrices are in the penalization region (LSTVAR~2) than the specification where the true AR matrices lie well inside the stability region (LSTVAR~1). However, the differences are quite small, and even in small samples, there does not appear to be notably issues with the estimation accuracy of the AR matrices lying inside the penalization region. 



\section{Details on the empirical application}\label{sec:adequacy}
The estimation of the models, residual diagnostics, and computation of the GIRFs are carried out with the accompanying R package sstvars \citep{sstvars}, which implements the introduced methods. The R package sstvars also contains the dataset studied in the empirical application. The upper bound for the joint spectral radius (see Appendix~\ref{sec:stat}), in turn, is computed with the JSR toolbox \citep{Jungers:2023} in MATLAB. 


To study whether our model adequately captures the autocorrelation structure of the data, we depict the sample autocorrelation functions (ACF) and cross-correlation functions (CCF) of the (non-standardized) residuals for the first $24$ lags in Figure~\ref{fig:resacf}. There is not much autocorrelation left in the residuals, but there are large correlation coefficients (CC) in the lag zero CCFs between the residuals of CPUI and EPUI. In addition, there are a few moderately sized CCs in the ACFs and CCFs of some of the residual series, but the vast majority of the CCs are small. Hence, our LSTVAR model appears to capture the autocorrelation structure of the data reasonably well.

To study our model's adequacy to capture the conditional heteroskedasticity in the data, we depict the sample ACFs and CCFs of squared standardized residuals for the first $24$ lags in Figure~\ref{fig:res2acf}. There are particularly large CC in the lag one CCFs between the IPI and EPUI, in the lag one CCF between the RATE and EPUI, and the lag zero CCF between RATE and EPUI. In addition, there are several sizeable CCs in the ACFs and CCFs of some of the residual series. Thus, there are clearly some inadequacies in capturing the conditional heteroskedasticity, but the vast majority of the CCs are small.

Finally, the quantile-quantile-plots based on the independent skewed $t$-distributions with mean zero, variance one, and degrees-of-freedom and skewness parameter values given by their estimates are presented in Figure~\ref{fig:serqq} for each residual series. They show that the marginal distributions of the series are reasonably well captured, although there is a single very large outlier in the IPI's residuals (from March 2020), which is related to the vast and fast drop in production caused by the COVID-19 lockdown. Therefore, the overall adequacy of the model seems, in our view, reasonable.

\begin{figure}[t]
    \centerline{\includegraphics[width=\textwidth - 2cm]{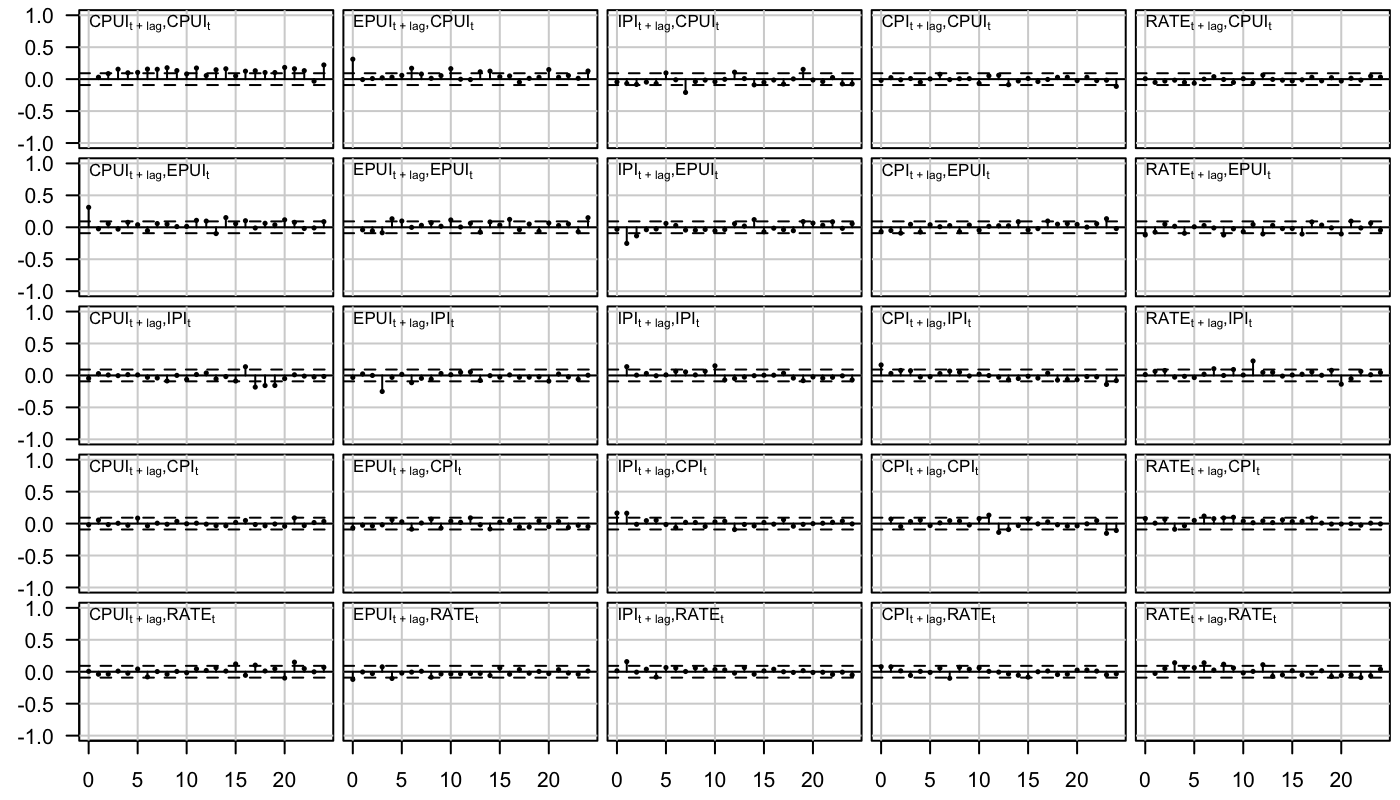}}
    \caption{Auto- and crosscorrelation functions of the residuals of the fitted two-regime second-order LSTVAR model for the lags $0,1,...,24$. The lag zero autocorrelation coefficients are omitted, as they are one by convention. The dashed lines are the $95\%$ bounds $\pm 1.96/\sqrt{T}$ for autocorrelations of IID observations.}
\label{fig:resacf}
\end{figure}

\begin{figure}[p]
    \centerline{\includegraphics[width=\textwidth - 2cm]{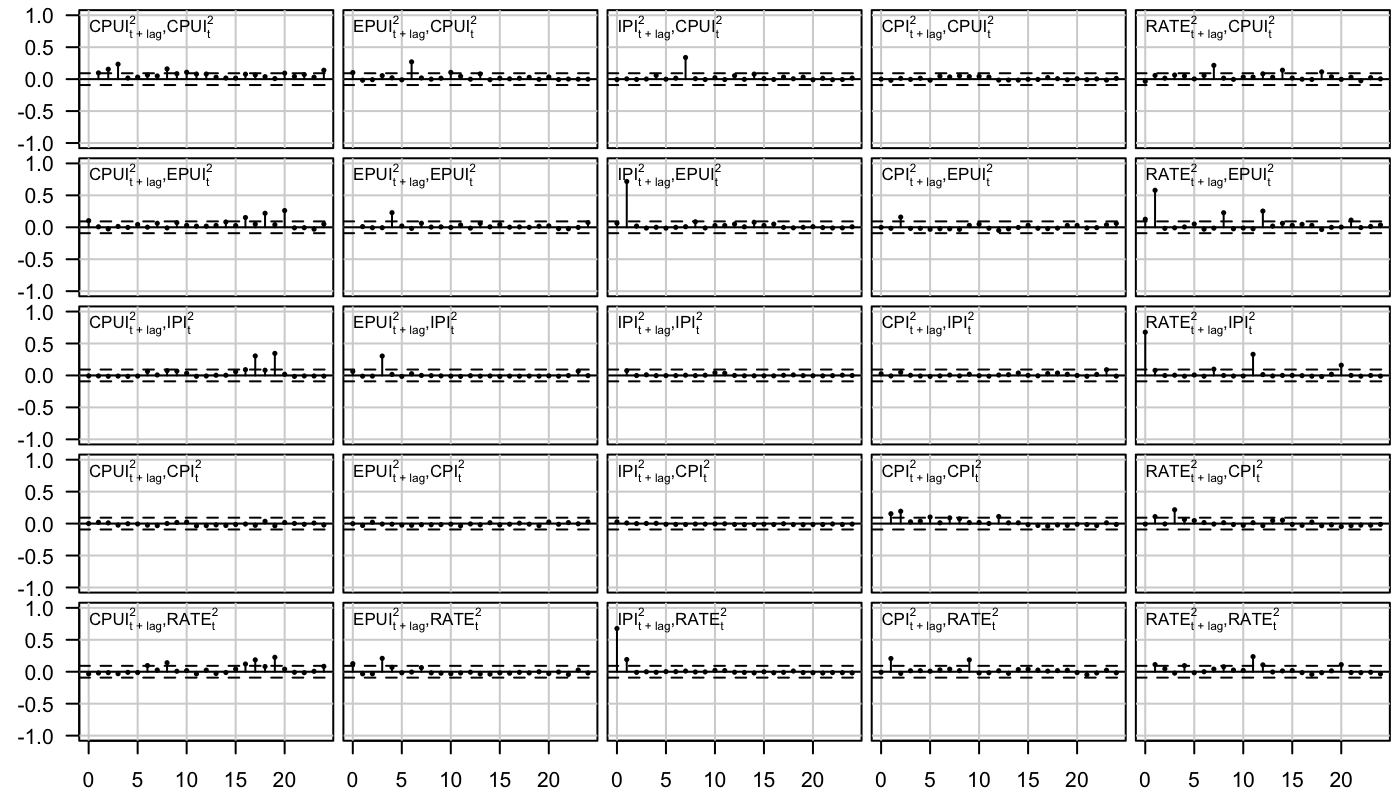}}
    \caption{Auto- and crosscorrelation functions of the squared standardized residuals of the fitted two-regime second-order LSTVAR model for the lags $0,1,...,24$. The lag zero autocorrelation coefficients are omitted, as they are one by convention. The dashed lines are the $95\%$ bounds $\pm 1.96/\sqrt{T}$ for autocorrelations of IID observations.}
\label{fig:res2acf}
\end{figure}

\begin{figure}[p]
    \centerline{\includegraphics[width=\textwidth - 2cm]{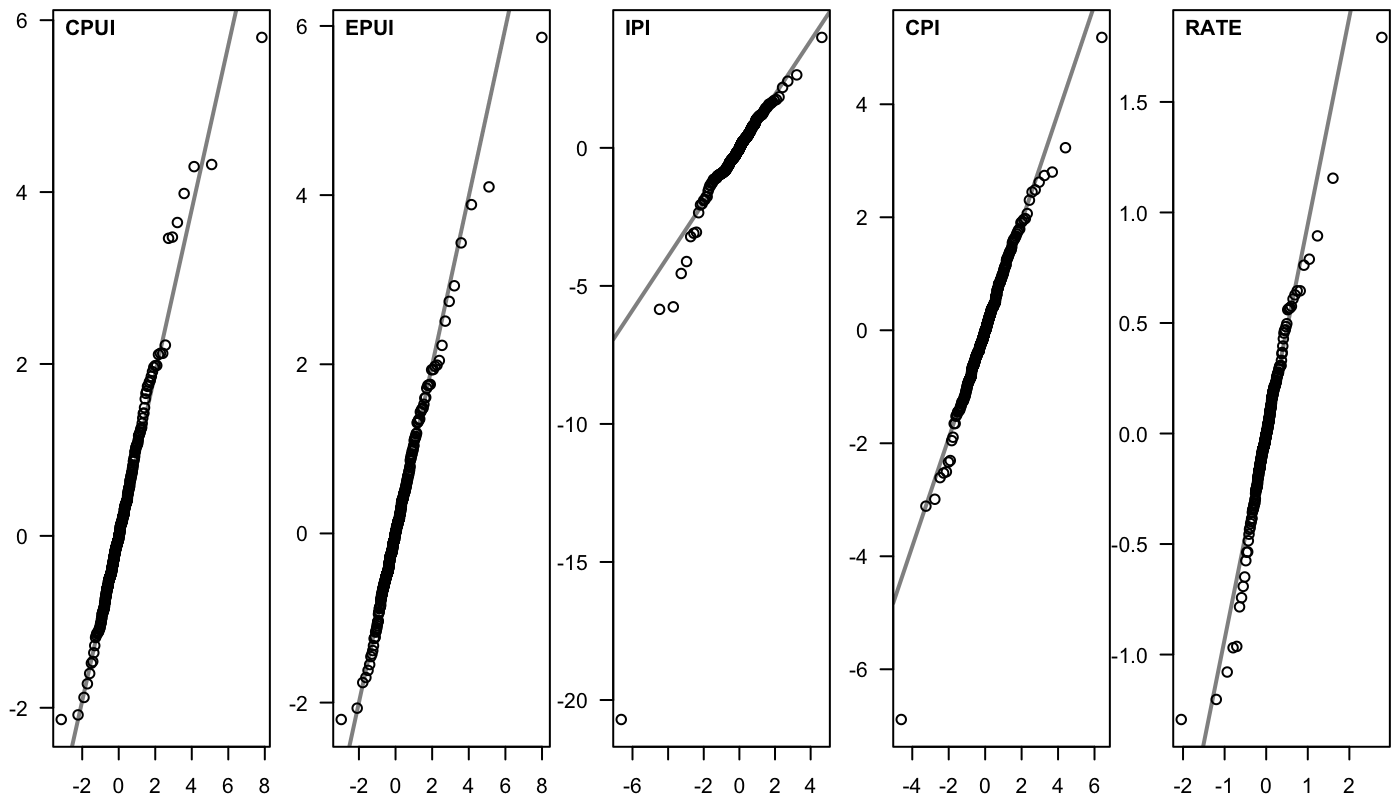}}
    \caption{Standardized residual time series and quantile-quantile-plots of the fitted two-regime second-order LSTVAR model. The quantile-quantile plots are based on $t$–distributions with zero mean, variance one, and degrees-of-freedom and skewness parameter values given by their estimates.}
\label{fig:serqq}
\end{figure}

\end{appendices}

\end{document}